\documentclass[12pt,draftclsnofoot,journal,onecolumn]{IEEEtran}
\usepackage{array}
\newcolumntype{P}[1]{>{\centering\arraybackslash}p{#1}}
\usepackage{amsmath}

\usepackage{floatrow}
\newfloatcommand{capbtabbox}{table}[][\FBwidth]

\usepackage{blindtext}
\usepackage{float}

\usepackage[justification=centering]{caption}
\usepackage[font=small,labelfont=bf]{caption}
\usepackage[english]{babel}
\usepackage{amsmath,amssymb}
\usepackage{times}
\usepackage{relsize}
\usepackage{graphicx}
\usepackage{subcaption}

\usepackage{url}
\usepackage{booktabs}
\usepackage{multirow}
\usepackage{mparhack}

\usepackage{authblk}
\usepackage{amsthm}

\usepackage{setspace}
\usepackage[noend]{algorithmic}
\usepackage[linesnumbered,ruled,vlined]{algorithm2e}
\usepackage{array}
\usepackage{cite}
\usepackage{eqparbox}
\usepackage{mdwmath}
\usepackage{balance}
\usepackage{epsfig}
\usepackage{xcolor}
\usepackage{amsthm}

\usepackage{mathtools}
\usepackage{bm}
\usepackage{amsfonts}
\usepackage[all]{xy}
\usepackage{etoolbox}
\usepackage{graphicx}
\usepackage{enumerate}
\usepackage{rotating}
\usepackage{pstricks}

\usepackage[format=hang]{caption}

\DeclareMathAlphabet{\mathantt}{OT1}{antt}{li}{it}
\DeclareMathAlphabet{\mathpzc}{OT1}{pzc}{m}{it}

\usepackage{accents}
\newtheorem{theorem}{Theorem}

\newtheorem{lemma}[theorem]{Lemma}

\DeclareFontFamily{OT1}{pzc}{}
\DeclareFontShape{OT1}{pzc}{m}{it}%
  {<-> s * [1.1] pzcmi7t}{}
\DeclareMathAlphabet{\mathpzc}{OT1}{pzc}%
                     {m}{it}

\def\U{\mathcal{U}}
\def\F{\mathcal{F}}
\def\V{\mathcal{V}}

\DeclareMathOperator{\argmin}{\arg\min}

\makeatletter
\patchcmd{\@maketitle}
  {}
  {}
\makeatother

\title{Routing and Scheduling of Network Flows with Deadlines and Discrete Capacity Allocation}

\author{
Ghafour Ahani,
Pawel Wiatr, and
Di Yuan

\thanks{G.\ Ahani, P.\ Wiatr, and D.\ Yuan are with the Department of Information Technology, Uppsala University, 751 05 Uppsala, Sweden (e-mail:
\{ghafour.ahani, pawel.wiatr, di.yuan\}@it.uu.se).}

 \vspace*{-2em}
}
\usepackage{graphicx}

\usepackage[justification=justified,singlelinecheck=true]{caption}
\begin{document}

\maketitle
\begin{abstract}

Joint scheduling and routing of data flows with deadline constraints
in communication networks has been attracting research interest.  This
type of problem distinguishes from conventional multicommodity flows
due to the presence of the time dimension.  In this paper, we address
a flow routing and scheduling problem with delivery deadline,
where the assignment of link capacity occurs in discrete units.
Discrete capacity allocation is motivated by applications in
communication systems, where it is common to have a base unit of
capacity (e.g., wavelength channel in optical communications).  We
present and prove complexity results of the problem. Next, we give an
optimization formulation based on a time slicing approach (TSA), which
amounts to a discretization of the time into time slices to enable to
formulate the deadline constraints. We then derive an
effective reformulation of the problem, via which a column generation
algorithm (CGA) is developed. In addition, we propose a simple and
fast Max-Flow based Algorithm (MFA). We use a number of network and
traffic scenarios to study various performance aspects of the
algorithms.

\end{abstract}

\section{Introduction}

Joint optimization of routing and scheduling of data flows across networks
is of importance to many applications, such
as data exchange in scientific projects~\cite{bandwidthprovision11},
data replication between datacenters~\cite{JamesHamiltonBlog}, as well as
reducing carbon footprint~\cite{noteasybeinggreen} and electrical expenses
~\cite{Qureshi2009} of datacenter networking.
A data flow may be subject to a
delivery deadline~\cite{Wilson2011}. For example, a flow may be a data backup or a
collected information from weather stations. The former is less time
sensitive and may happen in the background, while the latter is more
time sensitive and may be subject to a
deadline~\cite{ImprovedScheduling}. In fact, a survey of customers by
Microsoft~\cite{Jalaparti2016} reveals that most of them desire a
deadline guarantee for data delivery, and timely delivery carries
economic incentives as well~\cite{ji2018}.

We consider a flow routing and scheduling problem, where each flow is
characterized by a source, a destination, and amount of data (in bits)
to be delivered. The capacity allocated to a flow on a link
corresponds to the data transmission rate (in bits per second). A flow
may be subject to a deadline, before which the entire amount of demand
has to reach the destination. The objective is to minimize the overall
completion time. We consider, what was not studied extensively earlier,
that the capacity of a link can be allocated to flows only
in the form of discrete units; each unit is allocated to at most one
flow. This is rather common in communication networks, where
capacity allocation among flows has limitation in granularity. For example, in
optical networks the capacity of a link is often defined by the number
of transmission channels (wavelengths), and a wavelength channel is
not shared among flows.

The key difference between our problem setting and classic flow
problems (in particular multicommodity flow,
e.g.,\cite{larsson2004,Minoux2006,ahuja1993}) is the time
component and scheduling aspect because of the presence of deadline.
Some works (e.g.,\cite{ImprovedScheduling}, \cite{schedulingdeadline},
\cite{AdvanceReservations2009,multiplebulkdata,Liudong2015,GuaranteeingDeadlines2017,noor2017}),
which we will review in more details in
Section~\ref{sec:related_works}, have addressed the time aspect,
though not for discrete capacity allocation. The combination of
deadline constraint and discrete capacity allocation imposes challenges.

There are several objectives of our study. First, we would like to
examine if scheduling with deadline impacts problem complexity.
Second, our study targets developing mathematical formulations and
understanding to what extend they enable problem solution.  Another
key objective is to derive efficient solution algorithms in terms of
optimality and time efficiency. Finally, we aim at using computational
study to shed light on how well algorithms perform for application
scenarios of interest.

Our main findings, with respect to the above objectives, are as follows.

\begin{itemize}

\item As we prove later, the problem is NP-hard, not only because
of routing with discrete capacity allocation, but also inherently due
to the time scheduling aspect. Namely, the problem is NP-hard even if
the routing is fully fixed.

\item A linear integer formulation can be derived using time slicing, i.e.,
the time is partitioned into slices (aka time slots). If there are
too few slices, however the formulation does not represent exactly the
original problem, hence the solution is sub-optimal in general and the
formulation may be infeasible even if feasible solutions exist to the
original problem. Using many slices, on the other hand, increases the
problem size.

\item We derive an alternative formulation, in which the
possible routing patterns (with associated capacity allocations) are
the key components, and time durations of using these patterns are the
optimization variables.  As there are exponentially many routing
patterns, we develop a column generation algorithm. Applying the
algorithm until its termination gives the optimal solution.

\item We also consider an intuitive and fast heuristic algorithm
based on maximum flows. The algorithm uses the deadlines to set
priorities of the flows.

\item Via extensive numerical results, it is observed that the column generation
algorithm outperforms time slicing approach. The max-flow based
algorithm, although fast, has issues with optimality and feasibility.
However, we found that column generation can be combined with
the heuristic as well as with the time slicing model, to gain
significant reduction of computing effort.

\end{itemize}

The remainder of the paper is organized as follows. In
Section~\ref{sec:related_works}, we review the related works. In
Section~\ref{sec:problem_definition}, we provide problem definition,
illustrative examples, as well as complexity analysis for two cases.
In the first case, both routing and scheduling are subject to
optimization. In the second one, routing is fixed and only scheduling
is present.  In Section~\ref{sec:timeSlicing}, we formulate the
problem using the time-slicing approach, and highlight potential
drawbacks of this approach.  Next, we reformulate the problem using
routing patterns, and present our column generation algorithm in
Section~\ref{sec:cga}.  The heuristic algorithm is then presented in
Section~\ref{sec:ha}.  Performance evaluation is given in
Section~\ref{sec:performance_evaluation} followed by conclusions in
Section \ref{sec:concolusion}.

\section{Related works}
\label{sec:related_works}

Existing studies of scheduling flows in networks with deadlines can be
categorized into two groups. The first group is based on heuristic
algorithms that offer sub-optimal solutions.  For example, fair
sharing that divides the link capacity equally among flows, is known to
be far from optimal in terms of minimizing the overall completion
time~\cite{Bansal2001} and meeting the flow
deadlines~\cite{Sivaraman2001}.  Scheduling flows with respect to
their deadlines is known to minimize the number of late flows, whereas
sending flows with smaller sizes first minimizes the mean flow
completion time~\cite{FinishingFlows}.  The authors of
\cite{Fatiha2010} proposed an algorithm based on solving
a sequence of maximum flow problems.
As it is apparent, none
of such algorithms guarantees obtaining the optimal solution and they
may fail in obtaining a feasible solution, even though they generally
run fast. The second group of algorithms uses the idea of discretization of time
into slices. The problem in question is divided into several
interconnected sub-problems, where each sub-problem deals with one
time slice in which some amount of data is delivered.

Chen and Primet \cite{schedulingdeadline} investigated multiple bulk
data transfer (MBDT) subject to deadline constraints with given
routing paths.  The problem is formulated as bandwidth allocation over
time to minimize the network congestion factor. Within every time
slice, a bandwidth is reserved for each data transfer.  Rajah
et~al.~\cite{rajah2008} studied the problem of dynamic MBDT with
deadlines, and proposed a scheduling framework based on dividing time
into uniform slices. The problem is then formulated as a maximum
concurrent flow problem, using throughput as the performance
objective. Extending this work, in \cite{AdvanceReservations2009} the
authors proposed a non-uniform time slicing method and introduced
admission control along with scheduling algorithms to minimize the
request rejection ratio. Wang et~al.~\cite{multiplebulkdata} studied
the problem of MBDT with deadlines and time-varying link
capacity. Assuming the presence of delay tolerance, the problem
amounts to reducing the peak traffic load on links along the temporal
dimension via store-and-forward. Time slicing is used for problem
formulation. In \cite{Yao2015}, the authors studied MBDT in
inter-datacenter networks for backups and recovery purposes in natural
disasters. In this work, the destination (backup site) of a data flow
is unknown and is subject to optimization. The objective is to
minimize the time for the backup process. In \cite{Yassine2016}, the
authors investigated MBDT with soft and hard deadlines. They proposed
a bandwidth on demand broker model.  The works in \cite{Yao2015} and
\cite{Yassine2016} considered discrete capacity allocation using
an optical wavelength channel as the base unit. Zuo and Zho~\cite{ImprovedScheduling} studied the problem of
scheduling multiple bandwidth reservation requests on one reservation
path. Two performance parameters, the completion time and the duration
of scheduling individual flows, were considered.  The authors
of~\cite{ImprovedScheduling} proved that both problems are
NP-complete. The authors also presented improvements of two heuristic
algorithms previously given in \cite{Liudong2015}. The authors of \cite{GuaranteeingDeadlines2017}
studied dynamic inter-datacenter data transfers with guaranteed
deadlines. They used time slots to model the timeline, and proposed two methods
to determine whether or not a new coming flow can be
accommodated. The work in \cite{noor2017} considered transferring
data from one source to multiple destinations, which is related
to the Steiner tree problem.

As was mentioned earlier, the two key elements in the problem we study
are scheduling with deadline and integer capacity allocation. In
respect of these, the current literature presented above use either
heuristics or time slicing formulations to approach problem solutions.
To the best of our knowledge, there is a lack of studies investigating
problem complexity inherently connected to the scheduling element, and
alternative formulations than enable efficient computing of optimum.
These represent the literature gap that the current paper intends to
fill.

\section{Problem Definition and Complexity Analysis}
\label{sec:problem_definition}

\subsection{Problem Definition}

A network is modeled by a directed graph
$\mathcal{G}=(\mathcal{N},\mathcal{A})$ where $\mathcal{N}$ is the set
of $N$ nodes and $\mathcal{A}$ is the set of $A$ arcs.  The arc from
node $i$ to node $j$ is represented by $(i,j)$. The capacity of arc
$(i,j)\in\mathcal{A}$ is denoted by $c_{ij}$.  Denote by
$\mathcal{U}=\{u_1,u_2,\dots,u_{k}\}$ the available set of capacity
units. For each arc and flow, the amount of the arc capacity allocated
is restricted to be a nonnegative integer combination of the elements
of $\U$. For arc $(i,j)$, the total allocated capacity to the flows
may not exceed $c_{ij}$. The allocated capacity to a flow represents
the maximum rate at which the data of this flow can be transmitted on
the link.  The set of flows is denoted by $\mathcal{F} = \{1, \dots,
F\}$. Each flow $f$, $f \in\mathcal{F}$, requires an amount of data to be sent from an origin
node to a destination node. A flow may have a deadline before which
its entire amount of data must be delivered to the destination. A flow
$f$ is specified by a 4-tuple $(o_f, d_f, t_f, s_f)$ where $o_f$,
$d_f$, $t_f$, and $s_f$ denote the origin, destination, deadline, and
size (i.e., amount of data), respectively. For
convenience, in Appendix~\ref{sec:notation}, we provide a summary of
the notion used for problem definition as well as those used in the
algorithms in later sections.

We say a flow is scheduled, if some positive amount of rate for this
flow, from the origin to the destination (hereafter referred to as
end-to-end rate), is achieved by capacity allocation along one or
multiple paths.  In general, the end-to-end rate of a flow changes over time and this rate may be zero for some time periods (in
which some other flows are scheduled), after which the flow is
scheduled again to deliver more
data~\cite{AdvanceReservations2009,multiplebulkdata}.  In other words,
a flow may be scheduled in multiple but not necessarily consecutive
time periods.  In each time period, some portion of data is sent and
its amount depends on the duration of time period and the rate at
which flow is sent.  Moreover, a flow may be routed along multiple
paths with different rates.  Path selection and capacity allocation
generally differ from one time period to another. Therefore, the
duration of time period, the capacity allocation, as well as route
selection are to be jointly optimized. That capacity allocation has
to be some integral combinations of the available capacity units
restricts the flow solution to be what is commonly known as integer
flows.  Hence we call our problem the Integer Flow with Deadline
Problem (IFDP), in which routing and scheduling are strongly
intertwined.  The objective is to complete all the flows (aka makespan
in some other context), subject to capacity and deadline constraints.
We remark that, if a flow is scheduled in some non-consecutive time
periods, the intermediate period will always be utilized for other
flows, because our objective function is to minimize the
makespan. That is, there will not be any idle time period in the
optimal schedule.

\subsection{Illustrative Examples}

We provide an example of IFDP in
Figure~\ref{fig:exmapleDrawbackTSA}, for a simple network of three
nodes and three arcs.  All arcs have unit capacity. Each flow has one
possible path. It is apparent that, at any time, only one flow can be
scheduled due to discrete capacity allocation. It is easy to see
that the optimal schedule begins with flow A for 0.5 time units,
followed by flow B for 1.5 time units, and finally flow C for 1 time
unit. The overall completion time is 3. Note that if continuous
capacity allocation is allowed, the corresponding optimum would route
flows A, B, and C simultaneously, each with 0.5 units of capacity, for
1 time unit. Next, flows B and C are combined again using 0.5 capacity units
each, for a duration of 1 time unit. Flow C is then scheduled alone
for 0.5 time units. The completion time is 2.5.

\begin{figure}[h!]
  \begin{subfigure}[b]{0.4\textwidth} \centering
  \includegraphics[scale=0.5]{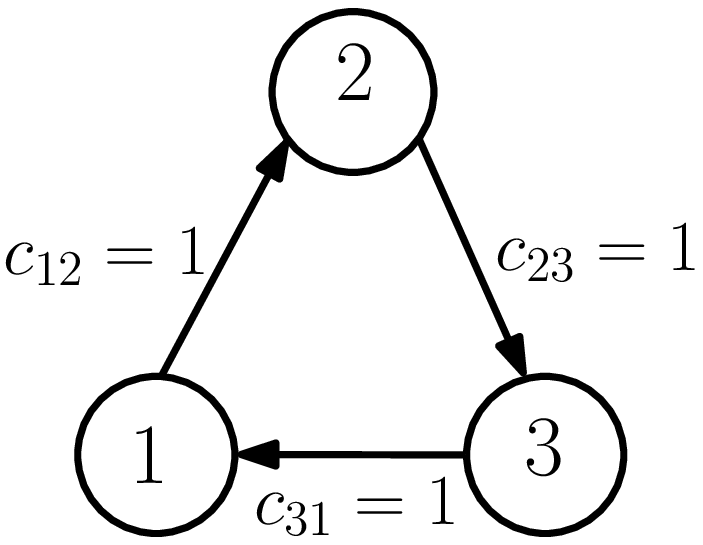}
\caption{Network topology.}
  \end{subfigure}
\begin{subfigure}[b]{0.5\textwidth}
\small
\centering
\begin{tabular}{ c c c c c }
\hline
\hline
Flow& Origin& Destination& Size & Deadline \\
\hline
\hline
A & 1 & 3 &0.5&1\\

B & 2 & 1 &1.5&2\\

C & 3 & 2&1&3 \\
\hline
\end{tabular}
\caption{Flow parameters.}
\end{subfigure}
  \caption{An example with 3 flows on a network with link capacity one.}
  \label{fig:exmapleDrawbackTSA}
\end{figure}

As can be seen from the example and its solution, the routing aspect
of IFDP clearly assembles integer flows applied to
telecommunication networks~(e.g., \cite{Minoux2006,BrCoFi2000}). On the
other hand, the scheduling aspect that takes place along the time
dimension, and the presence of deadlines make IFDP different from classical
flow problems.

\begin{figure}[h!]
  \begin{subfigure}[b]{0.4\textwidth} \centering
  \includegraphics[scale=0.5]{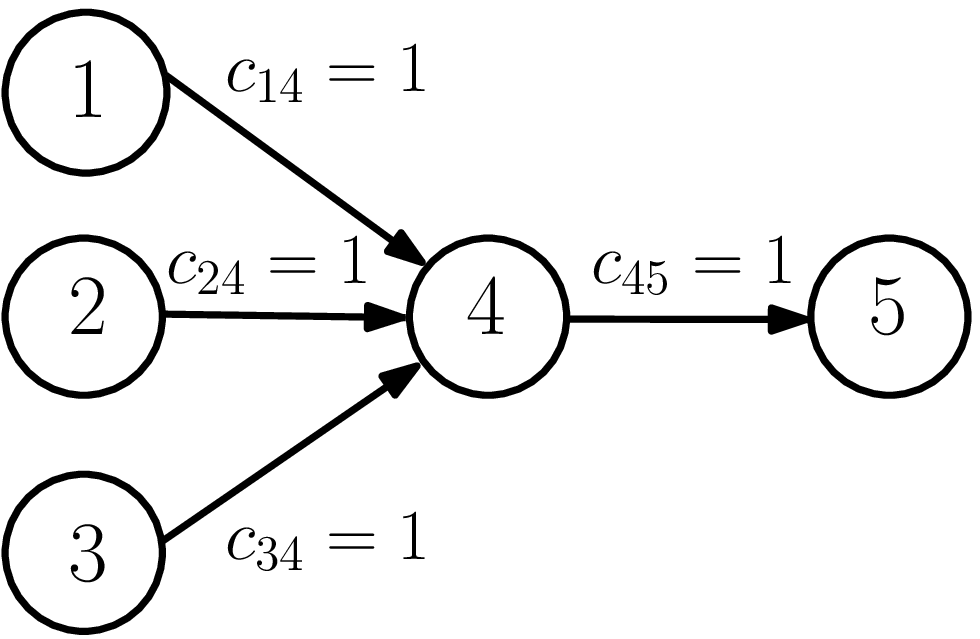}
\caption{Network topology.}
  \end{subfigure}
\begin{subfigure}[b]{0.5\textwidth}
\small
\centering
\begin{tabular}{ c c c c c }
\hline
\hline
Flow& Origin& Destination& Size & Deadline \\
\hline
\hline
A & 1 & 5 &1&1\\

B & 2 & 5 &1&2\\

C & 3 & 5 &1&3 \\

D & 2 & 4 &2&3 \\

\hline
\end{tabular}
\caption{Flow parameters.}
\end{subfigure}
  \caption{An example for which a flow
is not scheduled in consecutive time periods at optimum.}
  \label{fig:example2}
\end{figure}

We give a second example in Figure~\ref{fig:example2} to illustrate
non-consecutive scheduling of a flow along time.  As for the previous example,
all arcs have unit capacity. Consider first flows A, B, and
C. Because they share a common arc, the only feasible solution to meet
their deadlines is to schedule A, B, and C, for the given order using
one time unit each. The overall completion time is 3.
Now consider flow D, with size 2 and deadline 3. The path of flow D
is in conflict with that of flow B. However, flow D can be scheduled first
together with flow A in the first time unit, and then again with flow C
in the third time unit, meeting the deadline of flow D and the overall completion time
remains 3.

\subsection{Complexity Analysis}
\label{NPHard}

IFDP consists of jointly considering flow routing and scheduling.
Because multicommodity flow with integer capacity allocation is
NP-hard, even for the case of single size of capacity unit~\cite{Even74}, the
NP-hardness of IFDP is expected.  We formalize this result for the sake
of completeness.

\begin{theorem}
\label{th:NP1}
IFDP is NP-hard.
\end{theorem}
\begin{proof}
See Appendix~\ref{NPhard1}.
\end{proof}

A more interesting aspect of problem complexity arises when the path
selection is fixed. Note that multicommodity flow with integer capacity
allocation is no longer NP-hard, if one fixed path is given for each
commodity, as the problem will reduce to a simple feasibility check.
In the following, we show the fact that the complexity of IFDP is not
only due to integer capacity allocation, but also inherently connected
to scheduling along the time dimension. Namely, the problem remains
NP-hard, even if the paths of the flows are fully fixed.

\begin{theorem}\label{th:NP2}
IFDP with given paths is NP-hard.
\end{theorem}
\begin{proof}
See Appendix~\ref{NPhard2}.
\end{proof}

We end the section by a special case, for which IFDP can be solved in
polynomial-time. Namely, there is exactly one bottleneck link on which
all flows interact in terms of capacity sharing, which for example
appears in some application scenarios~\cite{ImprovedScheduling}. Note
that the paths to be used by the flows are not given, however all
candidate paths include the bottleneck link.
In this case, IFDP reduces to a
sequence of maximum flow problems.

\begin{theorem}\label{th:Poly}
If the possible paths of flows all share exactly one common link, of which the
capacity is non-redundant for any flow (i.e., the maximum end-to-end rate of
the flow equals the capacity, even if the other flows are discarded),
then the optimum of IFDP is to route and schedule the flows separately
in time, in ascending order of the deadlines.
\end{theorem}
\begin{proof}
See Appendix~\ref{Poly}.
\end{proof}

\section{Mathematical Formulation using Time-slicing}
\label{sec:timeSlicing}

\subsection{Problem formulation}

Mathematically formulating IFDP in not very obvious.  One
possibility of representing the problem is the Time-Slicing based
Approach (TSA).  In TSA, time is divided into time slices, of which
the set is denoted by $\mathcal{T}$, and the length of time slice $\tau$
is denoted by $|\tau|$. Each time slice is
associated with a capacity allocation. This allocation is used
throughout the time
slice. Hence this is an approximate formulation, and the accuracy
depends on the granularity of time slicing. For each time slice, the
problem is similar to regular multicommodity flows.  The time slices
are then considered jointly with respect to data demands and
deadlines. The concept of time slicing has been used for modeling
and solving problems related to IFDP, see \cite{schedulingdeadline,
AdvanceReservations2009,multiplebulkdata,rajah2008}. The TSA
we present below is an adaption of the concept to our problem.

We use variable $y^{\tau}_{fij}$ to denote the rate allocated to flow
$f$ on arc $(i,j)$ in time slice $\tau$.  Variable $r_f^{\tau}$ is
used to represent the total end-to-end rate of flow $f$ in time slice
$\tau$. Consequently, $|\tau|r_f^{\tau}$ is the amount of data of flow
$f$ that is delivered in this time slice.  Let $z_{fij}^{m,\tau} \in
\mathbb{Z}^+$ be a non-negative integer variable, denoting how many
times capacity unit $m$ is allocated to $f$ on $(i,j)$ in
$\tau$. Finally we use binary variable $w_\tau$ to denote if there is
any flow with positive rate in time slice $\tau$. The TSA formulation
is given below.

\begin{subequations}
\begin{alignat}{2}
\quad &
\min~~\sum_{\tau \in \mathcal{T}} |\tau| w_\tau \label{obj:TSA} \\
\text{s.t}. \quad
&
\sum_{\{j|(i,j) \in \mathcal{A}\}}y^{\tau}_{fij} - \sum_{\{j|(j,i) \in \mathcal{A}\}}y^{\tau}_{fji}=
\begin{cases}
  -r_f^{\tau},                & \text{if } i=o_f\\
    r_f^{\tau},              & \text{if } i=d_f,~\forall f \in \mathcal{F}, \forall i \in \mathcal{N}, \forall \tau \in \mathcal{T}\\
    0,                & \text{otherwise}
\end{cases}
\label{const:flowbalanceTSA}\\
&
\sum_{\{\tau \in \mathcal{T}:\tau \le t_f\}}  |\tau| r^{\tau}_f  = s_{f},\forall f \in \mathcal{F}
\label{const:sizeTSA}\\
&
 y^{\tau}_{fij} \leq \sum_{m=1}^{k} u_{m}
z_{fij}^{m,\tau},\forall f \in \mathcal{F}, \forall (i,j) \in
\mathcal{A}, \forall \tau \in \mathcal{T} \label{const:capacity1TSA}\\
&
\sum_{f \in F}\sum_{m=1}^{k}  u_{m} z_{fij}^{m,\tau} \le c_{ij}w_{\tau},
\forall (i,j) \in \mathcal{A}, \forall \tau \in \mathcal{T}  \label{const:capacity2TSA} \\
&
|\tau| r_{f}^\tau \le s_f w_\tau,\forall f \in \mathcal{F}, \tau \in \mathcal{T} \label{const:maxdataTSA} \\
&
w_{\tau} \leq w_{\tau-1},\forall \tau \in \mathcal{T}: \tau \geq 2 \label{const:Removing_symmetryTSA} \\
&
w_\tau \in \{0, 1\}, \forall \tau \in \mathcal{T} \\
&
y^{\tau}_{fij} \geq 0, \forall f \in \mathcal{F}, \forall (i,j) \in
\mathcal{A}, \forall \tau \in \mathcal{T} \\
&
z_{fij}^{m,\tau} \geq 0, \text{integer}, \forall f \in \mathcal{F}, \forall (i,j) \in
\mathcal{A}, \forall \tau \in \mathcal{T}, \forall m \in \{1,\dots,k\} \\
&
r_f^\tau \geq 0,  \forall f \in \mathcal{F}, \forall \tau \in \mathcal{T}
\end{alignat}
\label{eq:tsa}
\vskip -20pt
\end{subequations}

The objective function \eqref{obj:TSA} states the minimization of the total number
of used time slices.  The flow conservation constraints are formulated
in \eqref{const:flowbalanceTSA}.  The next set of constraints
\eqref{const:sizeTSA} state that the entire amount of data of a flow
has to be delivered in the time slices before the deadline. Inequality
\eqref{const:capacity1TSA} bounds the arc rate to what is permitted by
the values of $z$-variables and the right-hand side is a non-negative
integer combination of the capacity units.  Note that the inequality
allows the rate of flow on link to be smaller than the allocated
capacity. This is because we have formulated
\eqref{const:sizeTSA} using equality.
The arc capacity constraint is given in \eqref{const:capacity2TSA}.
It should be remarked that the formulation remains valid without the
presence of $w_\tau$ in the right-hand side. However the inclusion of
$w_\tau$ is for the purpose of strengthening the linear programming
(LP) relaxation. In \eqref{const:maxdataTSA}, the $w$-variables are
linked to the end-to-end rate variables, such that $w_\tau = 1$ if
there is any positive rate of any flow in time slice $\tau$, otherwise
$w_\tau = 0$ due to minimization. By \eqref{const:Removing_symmetryTSA},
time slices have to be used consecutively, starting with the first slice.
This, together with \eqref{obj:TSA}, implies that the optimal solution
does lead to the minimum overall completion time.

\subsection{Remarks on TSA}

As TSA is an integer linear model, standard integer programming
solvers can be used for its solution. One difficulty of using TSA is
that it is not obvious how many time slices should be defined.
Clearly, by increasing the number of time slices, all feasible
solutions of IFDP are eventually feasible solutions of TSA as well,
and hence solving TSA leads to the optimum of IFDP. A large number of
slices, however, significantly increases the size of the TSA model.  On
the other hand, if too few slices are used, solving the TSA model may
give a sub-optimal solution, and it may happen that no feasible
solution can be found at all via the model, even though such solutions
exist to the original IFDP instance.

Consider again the example given in Figure~\ref{fig:exmapleDrawbackTSA}.
Suppose time is divided into three time slices, each having
$|\tau|=1$, i.e., $[0,1]$, $(1,2]$, and $(2,3]$.  It is easy to see that
there is no feasible solution to the TSA formulation.  With
continuous capacity, however, the problem instance is feasible. The
optimal solution is illustrated in Figure~\ref{fig:example1}(b),
consisting of scheduling $0.5$ unit of each flow in time slice
$[0,1]$, followed by scheduling $1$ unit of flow $B$ using full
capacity of arcs $(2,3)$ and $(3,1)$ in time slice $(1,2]$, and
finally $0.5$ unit of capacity of arcs $(3,1)$ and $(1,2)$ to flow $C$
in the last time slice $(2,3]$. The overall completion time is
$3$. Note however this value originates from the granularity in time slicing;
in fact, the completion time would be shorter if the last time slice can be
broken into smaller segments.

Next, suppose we use a higher granularity and set $|\tau|=0.5$, resulting
in six time slices: $[0,0.5]$, $(0.5,1]$, $\dots$, and $(2.5,3]$.  In
this case, feasible solutions exist also with discrete capacity
allocation.  The optimum with discrete allocation is to schedule flow
$A$ in the first time slice, flow $B$ in the next three time slices,
and flow $C$ in the last two time slices. The deadlines are met, and
the solution has a completion time of $3$, see
Figure~\ref{fig:example1}(c). With continuous capacity allocation, an
optimal solution is illustrated in Figure~\ref{fig:example1}(d),
consisting of scheduling $0.5$ unit of each flow in the first two time
slices, followed by delivering $1$ unit of flow $B$ using full
capacity of arcs $(2,3)$ and $(3,1)$ in the next two time slices, and
finally allocating full capacity of arcs $(3,1)$ and $(1,2)$ to flow
$C$ in the next time slice. This gives a completion time of $2.5$.

\begin{figure}[h!]
\captionsetup{justification=centering}
\includegraphics[scale=1]{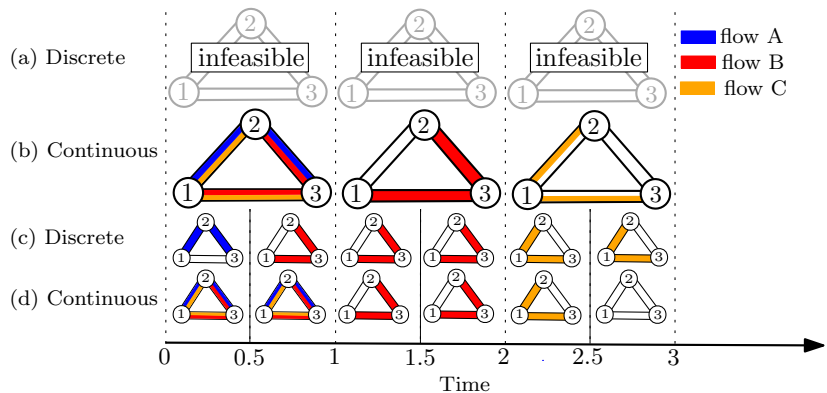}
\caption{Solutions with continuous and discrete capacity allocation.}
\label{fig:example1}
\end{figure}

From the above, if the number of time slices is too small, the
approach may fail even if there exist feasible solutions. In fact,
for the example network, using three time slices of uneven durations,
and, more specifically, $[0,0.5]$, $(0.5,2]$, and $(2,3]$
respectively, does lead to the feasibility as well as optimality.
However, in general there is no recipe for how to choose the slice
durations. Using a higher number of time slices will mitigate the
infeasibility issue, however at the cost of large problem size and
hence longer solution time. We also remark that, provided that the
original problem IFDP is feasible, one can easily prove that using
continuous capacity allocation leads always to a feasible solution
with any time slicing.  However it is rather hard to make use of the
solution to derive a solution to the original problem. Moreover, it
shall be remarked that, instead of using fixed time slices, one can treat
durations of the time slices as continuous variables. Doing so leads to a
nonlinear formulation, because the amount of data delivered
for a flow in a time slice is the product of the end-to-end rate,
which is a variable, and the duration of the slice.

Even though solving TSA may not be an effective approach to obtain a
high-quality solution for IFDP, formulation \eqref{eq:tsa} may be
useful for the purpose of bounding. In particular, consider a chosen
time length of $T$ and the LP relaxation of \eqref{eq:tsa} with one
single time slice of $[0, T]$. If the LP relaxation is infeasible,
then $T$ is a lower bound of the optimum of IFDP. The highest lower
bound can obtained via examining the largest $T$ (e.g., via bi-section
search) for which LP infeasibility remains. The bound can be used to
bound the optimality gap of a solution. Later in
Section~\ref{sec:performance_evaluation}, we study the effect of using
this bound together with a tolerance of optimality as a termination
criterion, on the solution time of our column generation method
presented in the next section.

\section{Problem reformulation and column generation}
\label{sec:cga}

We present a problem reformulation that enables a column generation
algorithm (CGA), that decomposes the problem into the so called master
problem (MP) and subproblem (SP). The algorithm iterates between a
restricted MP (RMP) and SP~\cite{viableCG}.  A key aspect in designing
a column generation algorithm is how to define the "column". In our
case, a column corresponds to an end-to-end rate vector of the flows,
resulted from a specific routing and capacity allocation solution.
More specifically, a rate vector is generically denoted by
$v=[r_1^v,r_2^v,\dots,r_F^v]^T$, in which element $r_f^v$, $f
\in \mathcal{F}$, represents the end-to-end rate at which flow $f$ is
sent from its origin to destination. The set of all possible vectors
is denoted by $\mathcal{V}$, of which the cardinality is clearly
finite.  For later use, we also define a subset of $\mathcal{V}$
denoted by $\mathcal{V}_f$ which consists of those vectors in
$\mathcal{V}$ for which the rate of flow $f$ is strictly
positive, i.e., $\mathcal{V}_f=\{v \in \mathcal{V} |r_f^v>0\}$. Note that
any feasible rate vector must satisfy flow conservation, flow rate, and
arc capacity constraints. These constraints will be utilized in the SP
for constructing new columns.

\subsection{MP and RMP}
For rate vector $v \in \V$, we define a continuous, non-negative
variable $x_v$ that represents how long time rate vector $v$ is used
in a problem solution. Also, without loss of generality, we assume that
the flow indices follow the ascending order of deadlines. The full MP
is as follows.

\begin{subequations}
\begin{alignat}{2}
\quad &
\min~~\sum_{v \in \mathcal{V}} x_v
\label{eq:mp} \\
\text{s.t}. \quad
&
\sum_{v \in \mathcal{V}_f} r_f^v ~ x_v = s_f,\forall f \in \mathcal{F} \label{const:sizeCGA}\\
&
\sum\limits_{v \in \mathcal{V}_1 \cup \dots \cup  \mathcal{V}_f }x_v  \leq t_f,\forall f \in \mathcal{F} \label{const:deadlineCGA}\\
&
x_v \ge 0, \forall v \in \V
\end{alignat}
\label{eq:reformulate}
\vskip -20pt
\end{subequations}

The objective function is to minimize the total time used by the rate
vectors. Constraints (\ref{const:sizeCGA}) state that for each flow,
the total amount of data sent from the source and received at the
destination equals its specified size. Constraints
(\ref{const:deadlineCGA}) are formulated for deadlines.

We remark that, while addressing the deadline requirements, the order
of the rate vectors in the schedule is of significance. Indeed, at
the first glance, it appears that using only the $x$-variables is not
sufficient -- we need also variables indicating the position of each
vector in the scheduling solution.  In the following, we show that
this is not necessary. Namely, the $x$-variables together with
(\ref{const:deadlineCGA}) which has the effect of partial ordering,
correctly consider the deadline constraints. To this end,
let us first consider the derivation of a schedule
based on the values of the $x$-variables of the MP.
This is necessary, because the
$x$-variables themselves do not carry any information of how the
rate vectors should be sequenced.

We present Algorithm~\ref{alg:Constructing_schedule_from_MP} that
derives a schedule using the $x$-variables. The input set to the
algorithm, $\mathcal{V}^+$, consists of rate vectors for which the
$x$-variables have strictly positive value at the optimum of the
RMP. As there are of $O(F)$ constraints in the LP, the number of such
rate vectors is of $O(F)$. The algorithm considers each vector $v$ in
$\mathcal{V}^+$, and finds the smallest flow index $f$ with positive
rate $v_f^v$.  The vector is then added to the corresponding set
$\mathcal{Q}_f$.  When all the vectors in $\mathcal{V}^+$ have been
examined, the algorithm goes through sets $\mathcal{Q}_1, \dots,
\mathcal{Q}_F$, to retrieve the vectors and their respective
time durations for problem solution. Notation $p_f$ in the algorithm
represents a time point by which flow $f$ is surely completed.

\begin{algorithm}
\caption{Constructing a schedule from MP}
\label{alg:Constructing_schedule_from_MP}
\begin{algorithmic}[1]
\REQUIRE $\mathcal{V}^+$
\STATE$p_0\leftarrow 0$, $\mathcal{Q}_f \leftarrow \emptyset,~\forall f \in \mathcal{F}$
\FOR{$v \in \mathcal{V}^+$}
\STATE $f^+ \leftarrow \min\{f \in \mathcal{F}:  r_f^v > 0\}$
\STATE $\mathcal{Q}_{f^+} \leftarrow \mathcal{Q}_{f^+} \cup \{v\}$
\ENDFOR
\FOR{$f=1:F$}
\STATE Schedule rate vectors $v \in \mathcal{Q}_f$ with the respective time durations
\STATE $p_f \leftarrow p_{f-1} +\sum\limits_{v \in \mathcal{Q}_f}x_v$
\ENDFOR
\end{algorithmic}
\end{algorithm}

The complexity of Algorithm~\ref{alg:Constructing_schedule_from_MP} is
determined by its first for-loop, since the second for-loop simply
retrieves the solution.  As the size of $\mathcal{V}^+$ is at most $F$,
the first for-loop goes though no more than $F$ vectors. For each of
them, determining the flow index by the min-operator clearly has a
complexity of $O(F)$. Hence the overall algorithm is of complexity
$O(F^2)$.

\begin{lemma}\label{lem:correctnessMP}
For any input representing a feasible solution of the MP, the
output of Algorithm \ref{alg:Constructing_schedule_from_MP}
corresponds to a feasible schedule.
\end{lemma}
\begin{proof}
Consider any flow $f$, $f \in \mathcal{F}$.
According to
the definition of $\mathcal{Q}_f$, flow $f$ is not scheduled
after time point $p_f$, because after this time point
there is no vector with positive rate for flow $f$.  Thus, the
overall amount of data delivered for $f$ is
$\sum_{i=1}^{f} \sum_{v \in \mathcal{Q}_i} r_f^v ~ x_v$. This is equal
to $\sum\limits_{v \in \mathcal{V}_f} r_f^v ~ x_v $ which in turn
equals $s_f$ by (\ref{const:sizeCGA}).
Denote the completion time of flow $f$ by $\chi_f$.
The total time for scheduling flows in
range $[1,f]$ is $\sum\limits_{v \in \mathcal{Q}_1 \cup \dots \cup
\mathcal{Q}_f} x_v$. We have
$\chi_f \le \sum\limits_{v \in \mathcal{Q}_1 \cup \dots \cup
\mathcal{Q}_f} x_v=\sum\limits_{v \in \mathcal{V}_1\cup \dots \cup
\mathcal{V}_f} x_v $ which is no greater than $t_f$ by
(\ref{const:deadlineCGA}). Note that the last equality
follows from that $\mathcal{Q}_f =\{v \in \mathcal{V}_f | r_i^v=0 \text{~for~} i \in \{1,..,f\} \& x_v>0\}$.
\end{proof}

What still remains is to show that (\ref{const:deadlineCGA})
itself is correct, i.e., in the sense that if IFDP is feasible,
then there must exist a solution satisfying (\ref{const:deadlineCGA}).

\begin{lemma}\label{lem:correct}
If IFDP is feasible, then there exists an $x$-solution satisfying
(\ref{const:deadlineCGA}).
\end{lemma}

\begin{proof}
Clearly, a solution to IFDP includes a set of rate vectors as well as
their time durations. The latter are the $x$-variables in
\eqref{eq:mp}. Assume feasibility, that is, there is some ordering of
these rate vectors, such that the sequence together with the values of the
corresponding $x$-variables satisfy all the deadlines.  Now suppose
$f' < f$, and in the sequence, a rate vector $v$ that has positive
rate for flow $f$ but zero rate for all flows $1, \dots, f'$, is
scheduled before a rate vector $v'$ with positive rate for flow
$f'$. Because all deadlines are met, $t_{f'}$ is not reached by the
end of the use of $v'$. Consider updating the solution, by putting
$v$ immediately after $v'$ instead.  Thus $v'$, along with
all vectors scheduled between $v$ and $v'$ in the original
sequence, are shifted earlier in time by $x_{v_f}$. Now the ending time of
vector $v$ equals that for $v'$ before the update. Because
$t_{f'} \leq t_f$, the deadline of $f$ remains satisfied.
The deadlines of the other flows are clearly also satisfied.
Repeating if necessary, it is apparent that after a finite number
of updates, the new $x$-solution satisfies (\ref{const:deadlineCGA}),
and the result follows.
\end{proof}

By the above result, the MP formulation is indeed a correct
mathematical model of IFDP.  Thus, even if there is an inherent
timeline in IFDP, the $x$-variables are sufficient for formulating the
deadlines.

The RMP is a restricted version of MP, such that $\V_f$ is replaced by
a subset $\V'_f$, with $|\V^{\prime}_f| \ll |\V_f|$, $f \in \F$. Let $\V' = \V'_1 \cup
\dots \cup \V'_F$. Apart from this difference, MP and RMP have the same
objective function and constraints. Hence we do not write out RMP in
its full form to save space. Also, for convenience, we will use
\eqref{const:sizeCGA} and \eqref{const:deadlineCGA}
to refer to the constraints of RMP with the restricted sets, as long
as this does not lead to ambiguity. When constructing RMP initially,
it is preferable that $\V'$ results in a feasible solution. This is
however not really a crucial issue, as one can apply a penalty-like phase
in case of initial infeasibility.

\subsection{SP Formulation}

After solving the RMP, we need to determine whether the current solution
is optimal.  As the MP is an LP, this amounts to finding an
unconsidered rate vector with negative reduced cost \cite{viableCG}.
The task is accomplished via solving an SP, of which the solution
is an end-to-end rate vector, resulted from routing and capacity allocation.
For SP, we reuse the notation $y$ and $z$ for variables.  The
difference to Section~\ref{sec:timeSlicing} is that there is no time
slice index for these variables here. The end-to-end rate of flow $f$
is denoted by variable $r_f$.

Denote by $\lambda_f^*$ and $\pi_f^*$ the optimal dual variable
values of \eqref{const:sizeCGA} and \eqref{const:deadlineCGA} of flow
$f$ in the RMP. Note that, by the structure of
\eqref{const:deadlineCGA}, $\pi_f^*$ appears in the reduced cost if
and only if the rate vector to be generated has strictly positive rate
for any of the flows $1, \dots, f$.  Hence, if we know $f^+$ is the
first flow index with positive rate in the generated rate vector, the
SP can be formulated as follows, where the objective function together
with the constant term $-\sum_{f = f^+}^F \pi_f^*$ corresponds to
reduced cost.

\begin{subequations}
\begin{alignat}{2}
\quad &
\min~~1-\sum_{f=f^+}^F r_f \lambda_f^*
 \\
\text{s.t}. \quad
&
\sum_{\{j|(i,j) \in \mathcal{A}\}}y_{fij}-\sum_{\{j|(j,i) \in \mathcal{A}\}}y_{fji}=
\begin{cases}
  -r_f,                & \text{if } i=o_f\\
    r_f,              & \text{if } i=d_f,~\forall f \in \mathcal{F}, \forall i \in \mathcal{N}\\
    0,                & \text{otherwise}
\end{cases} \label{const:flowbalanceCGA} \\
&
y_{fij} \le \sum_{m=1}^{k} u_{m} z_{fij}^{m},\forall f \in \mathcal{F}, \forall (i,j) \in \mathcal{A} \label{const:capacity1CGA} \\
&
\sum_{m=1}^{k}\sum_{f \in \F}  u_{m} z_{fij}^{m} \le c_{ij}, \forall (i,j) \in \mathcal{A} \label{const:capacity2CGA} \\
&
r_f\ge0, f\ge f^+ \label{const:rf1}\\
&
r_f =0, f< f^+ \label{const:rf2}\\
&
y_{fij} \geq 0, \forall f \in \mathcal{F}, \forall (i,j) \in
\mathcal{A} \\
&
z_{fij}^{m} \geq 0, \text{integer}, \forall f \in \mathcal{F}, \forall (i,j) \in
\mathcal{A}, \forall m \in \{1,\dots,k\}
\end{alignat}
\label{eq:subproblem}
\vskip -20pt
\end{subequations}

In the formulation, \eqref{const:flowbalanceCGA},
\eqref{const:capacity1CGA}, and \eqref{const:capacity2CGA} deal with
flow conservation, capacity allocation, and capacity limit. These
constraints are the counterparts of (\ref{const:flowbalanceTSA}),
(\ref{const:capacity1TSA}), and (\ref{const:capacity2TSA}),
respectively. The other constraints state the variables domains.
Among them, \eqref{const:rf1} and \eqref{const:rf2} set the domains of
the rate values with respect to index $f^+$.

If at the optimum $r_f^*, f \in \{f^+,\dots, F\}$, $1-\sum_{f=f^+}^F r^*_f \lambda_f^* - \sum_{f =
f^+}^F \pi_f^*$, which is the most negative reduced cost among all
possible rate vectors, is negative, we add the optimal rate vector,
$v^* = [r_1^{*},\dots,r_F^{*}]^T$, to RMP, and move to the next
iteration. Otherwise, the current solution via the RMP is optimal.

A solution of \eqref{eq:subproblem} does not necessarily have strictly
positive rate for $f^+$. This is because one cannot enforce something
to be positive but the amount to be arbitrary in optimization. Hence,
to be more precise, $f^+$ is the first flow index of which the rate
may be positive.  Note that $-\pi_f, f \in \{1, \dots, F\}$, are non-negative,
hence they are costs in minimization.  In \eqref{eq:subproblem},
the costs of allowing positive rates for flows $f^+, \dots, F$ are taken
into account and hence these flows may be allocated positive rates,
whereas for flows $1, \dots, f^+ -1$, the costs are excluded and
they must have zero rate. We also remark that, if it turns
out that the first index of positive flow is $f^{++} > f^+$,
the correct reduced cost equals $1-\sum_{f=f^+}^F r^*_f \lambda_f^* - \sum_{f =
f^{++}}^F \pi_f^*$ as $-\pi_f, f \in \{f^+, \dots f^{++}-1\}$ have to be omitted.
Note that in the first sum, the first index can be either $f^+$ or $f^{++}$
without affecting its correctness.

\subsection{Solving the SP}

That $f^+$ is in fact unknown can be addressed by solving a sequence
of SPs with different values of $f^+$. Intuitively, one can set $f^+$
to $1, \dots, F$ and solve the SP exactly $F$ times.  In the
following, we show that better efficiency may be achieved by observing
the first positive flow in the solution, while going through the
sequence.

Denote by $\text{SP}_{f^+}$ the SP for given index $f^+$ indicating
the first flow that may have positive rate.  Thus flows $1,\dots, f^+
- 1$ have zero rate, and flow $f^+, \dots, F$ have non-negative rate
in the solution of $\text{SP}_{f^+}$. Suppose that in the solution of
$\text{SP}_{f^+}$, $f^{++} \geq f^++1$ is the first flow having positive
rate. Hence the current objective function value is
$1-\sum\limits_{f=f^{++}}^{F}r^*_f
\lambda_f^*-\sum\limits_{f=f^{++}}^F \pi_f^*$.  The following lemma
states that this value cannot be improved by solving $\text{SP}_{f^+ +
1},
\dots, \text{SP}_{f^{++}}$.

\begin{lemma}\label{lem:skipping}
If after solving $\text{SP}_{f^+}$, the first flow with positive rate is flow $f^{++} \geq f^+ + 1$,
then  $\text{SP}_{f^+ +1}, \dots \text{SP}_{f^{++}}$ can be discarded without loss of optimality.
\end{lemma}

\begin{proof}
Denote by $\theta_{f^+}$ the most negative reduced cost obtained by
solving $\text{SP}_{f^+}$.  For the statement in the lemma,
$\theta_{f^+} = 1-\sum\limits_{f={f^{++}}}^{F}r^*_f
\lambda_f^*-\sum\limits_{f={f^{++}}}^F
\pi_f^*$.  Consider any $m \in [f^+ + 1, f^{++}]$. Any feasible solution of
$\text{SP}_m$ is also feasible in $\text{SP}_{f^+}$. This is because
the first flow that may be positive in the former is $m$, and
$m>f^+$. In other words, the solution space of $\text{SP}_{f^+}$ is
greater. Moreover, the optimum of $\text{SP}_{f^+}$ is feasible in
$\text{SP}_{m}$ because $m \leq f^{++}$. Therefore this solution must
also be optimal to $\text{SP}_{m}$. Hence $\theta_{m} = \theta_{f^+}$,
and the lemma follows.
\end{proof}

Based on what has been derived thus far, we present the procedure of
solving the SP in Algorithm~\ref{alg:SP}.
Here $r_f^*$ shall be understood and short-hand notation of
$r_f^{v^*}$. Also, note that, even if in the algorithm
presentation, one rate vector is produced by the end, in the process
multiple vectors with negative reduced costs (but with different flow
indices as the first positive element) may be produced. These can all
be added to the RMP, potentially speeding up column generation.

\begin{algorithm}
\caption{Algorithm for solving the SP}
\label{alg:SP}
\begin{algorithmic}[1]
\REQUIRE $\pi_f^*,\lambda_f^*, ~\forall f \in \mathcal{F}$
\ENSURE ${v^*}$
\STATE $f^+ \leftarrow 1$, $\text{optimum} \leftarrow \infty$, $v^* \leftarrow \emptyset$, $f^* \leftarrow 0$
\WHILE{$f^+\le F$}
\STATE Solve $\text{SP}_{f^+}$ to obtain its optimal rate vector $v^*$
\STATE $f^{++} \leftarrow \min \{f \in \{f^+, \dots, F | r^*_f >0 \}\} $
  \IF {($1-\sum\limits_{f=f^{++}}^{F}r^*_f \lambda_f^*-\sum\limits_{f=f^{++}}^F \pi_f^*< \text{optimum}$)}
       \STATE $\text{optimum} \leftarrow$ $1-\sum\limits_{f=f^{++}}^{F}r^*_f \lambda_f^*-\sum\limits_{f \ge f^{++}} \pi_f^*$
       \STATE ${v}^* \leftarrow [0,\dots,0,r_{f^{++}}^*,\dots,r_F^*]^T$, $f^* \leftarrow f^{++}$, $f^+ \leftarrow f^{++}+1$
   \ENDIF
\ENDWHILE
\end{algorithmic}
\end{algorithm}

\subsection{Summary of CGA}

We provide a summary of CGA in an algorithmic form in Algorithm~\ref{alg:frameworkCGA}.
The correctness of CGA in terms of optimality is then formally stated.

\begin{algorithm}
\caption{CGA}
\label{alg:frameworkCGA}
\begin{algorithmic}[1]
\STATE Start with an initial set of rate vectors $\V'$
\REPEAT
\STATE Solve the RMP
\STATE Compute $v^*$ and $f^*$ via Algorithm \ref{alg:SP}
\IF {($1-\sum\limits_{f=f^*}^F r_f \lambda_f^*- \sum\limits_{f=f^*}^F \pi_f^* \}< 0$)}
    \STATE $\V' \leftarrow \V' \cup \{v^*\}$
\ENDIF
\UNTIL {($1-\sum\limits_{f=f^*}^F r_f \lambda_f^*- \sum\limits_{f=f^*}^F \pi_f^* \} = 0$)}
\end{algorithmic}
\end{algorithm}

\begin{theorem}\label{th:CGA}
CGA solves IFDP to global optimality within a finite number of steps.
\end{theorem}
\begin{proof}
The results follows directly from LP optimality, that there are a finite number of rate vectors,
as well as Lemmas~\ref{lem:correctnessMP}-\ref{lem:skipping}.
\end{proof}

There is no unique way of performing the first line in
Algorithm~\ref{alg:frameworkCGA}.  An obvious choice is to apply a
scheme similar to phase one of LP, namely, to introduce additional,
so called artificial variables to represent the amount of
demand not satisfied, and minimize the total unsatisfied demand.
For this reason, we will use phase~I to refer to the first step
of CGA.
Moreover, recall that SP is an integer multicommodity problem that is NP-hard,
hence from a theoretically viewpoint, Algorithm~\ref{alg:SP} is of
exponential time complexity. This holds also for
Algorithm~\ref{alg:frameworkCGA} as it uses Algorithm~\ref{alg:SP}
repeatedly. Numerically, however, CGA performs quite well in
scalability, as will be shown later in
Section~\ref{sec:performance_evaluation}.

\section{Max-Flow based Algorithm}\label{sec:ha}

In this section, we present a heuristic, referred to as Maximum-Flow
based Algorithm (MFA). MFA is different form the classic maximum
flow in the sense that we maximize the end-to-end rate of all
flows. Hence, it is a type of multicommodity flow.

The rationale of MFA is rather intuitive. Namely, the flows are
prioritized with respect to their deadlines.  The priorities are
represented using weights.  The sum of weighted rates of flows is
maximized. Then, the corresponding rate vector is used, until one of
the flows becomes completed. Next, the rates of completed flows are
set to zeros.  The process then repeats for the remaining flows with
updated demand size.  Mathematically, the optimization problem can be
formulated as follows, where $\frac{1}{t_f^2}$ is used as an example
of the weight for flow $f$. In the following, we use $\bf{r^*}$ to
denote the optimum rate vector of the formulation.

\begin{subequations}
\begin{align}
\text{$\text{MFA}$:}
\quad & \max \sum_{f \in \mathcal{F}} \frac{1}{t_f^2} r_f \quad & \\
\qquad & \text{s.t.} \ \
\text{(\ref{const:flowbalanceCGA})},\text{(\ref{const:capacity1CGA})},\text{(\ref{const:capacity2CGA})} \nonumber
\end{align}
\end{subequations}

\begin{algorithm}
\caption{Max-Flow based Algorithm}
\label{alg:MFA}
\begin{algorithmic}[1]
%\REQUIRE network $\mathcal{G}=(\mathcal{N},\mathcal{A})$, set of flows with sorted deadlines in ascending order
%\ENSURE solution for the problem
\STATE $\chi \leftarrow 0$
\WHILE {($\exists f, f \in \mathcal{F}$, with $s_f>0$)}
\STATE Fix $r_f=0$ if $s_f=0$, $f \in \mathcal{F}$ \label{rateto0}
\STATE Solve $\text{MFA}$ and obtain optimum rate vector $\bf{r^*}$
\STATE $\delta^{*} \leftarrow  \min_{f\in\{1,\dots,F:r^*_f>0\}} s_f/r^*_f $
\STATE $f^{*} \leftarrow  \argmin_{f\in\{1,\dots,F:r^*_f>0\}} s_f/r^*_f $
\IF {$\chi+\delta^*\le t_{f^*}$}
 \STATE $s_f \leftarrow s_f - r^*_f \delta^*, f \in \mathcal{F}$
 \STATE$\chi\leftarrow \chi+\delta^*$
\ELSE
\STATE Return "No~Solution"
\ENDIF
\STATE Return $\chi$
\ENDWHILE
\end{algorithmic}
\end{algorithm}

MFA is formalized in Algorithm~\ref{alg:MFA}. For the completed flows
(i.e., flows with zero demand size), their rates are set to zeros in
Line~\ref{rateto0}. For the uncompleted flows, Line $4$ computes
maximum weighted rate vector. Next, the flow with minimum (remaining)
time necessary to be completed is computed by Lines $5$-$6$.  Whether
this flow can be delivered within its deadline or not is checked in
Line $7$. If not, MFA fails in finding a feasible solution. Otherwise,
the demand size and overall completion time are updated in Lines $8$
and $9$, respectively.

In addition to acting as a fast heuristic for large-scale instances of
IFDP, MFA can be used to speed up CGA. Namely, if a feasible solution
is found by MFA, this solution can be used to initialize the columns
in CGA, thus eliminating the need of phase one of LP in the context of
CGA. In the next section we will show significant improvements in
solution time by the combination of MFA and CGA.

\section{Performance Evaluation}
\label{sec:performance_evaluation}

\subsection{Scenario Setup}

We consider three network topologies with different sizes, depicted in
Figures~\ref{fig:SmallNetwork}, \ref{fig:SoftNetwork}, and~\ref{fig:GeantNetwork}, respectively.
The first is a small network composed of $6$ nodes and $8$ bidirectional links. The
second network topology is of Softlayer Inc
\cite{softlayerNetwork}, consisting of $11$ nodes and $17$
bidirectional links. The last topology is the
Geant network \cite{geantnetwork} that consists of $22$ nodes and $36$ bidirectional
links. In all networks, the capacity of each link is set to $10$ units
in each direction, and the unit in capacity allocation is $2$.

\begin{figure}[ht!]
\RawFloats
\begin{minipage}{1\textwidth}
\centering
\includegraphics[scale=0.8]{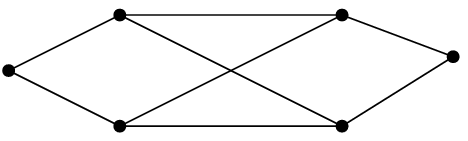}
\begin{center}
\caption{Small network with $6$ nodes and $8$ bidirectional links.}\label{fig:SmallNetwork}
\end{center}
\end{minipage}

\begin{minipage}{1\textwidth}
\centering
\includegraphics[scale=0.8]{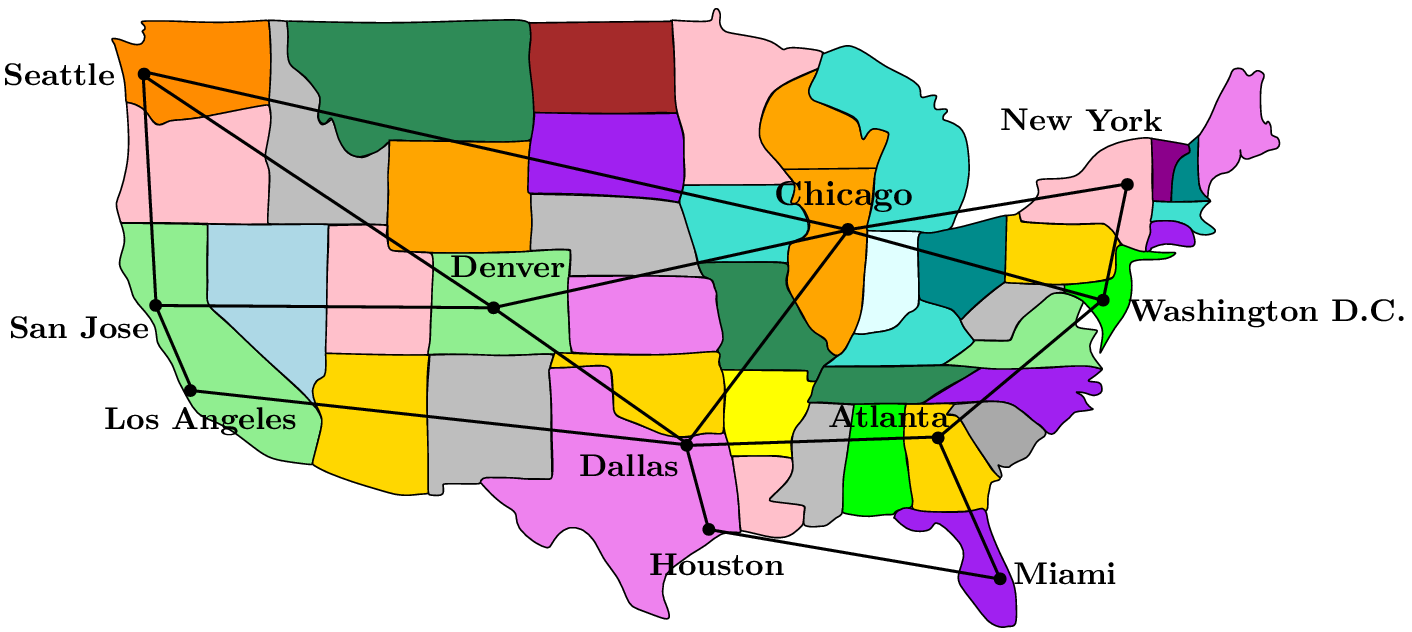}
\caption{Softlayer network with $11$ nodes and $17$ bidirectional links.}\label{fig:SoftNetwork}
\end{minipage}

\begin{minipage}{1\textwidth}
\centering
\includegraphics[scale=0.8]{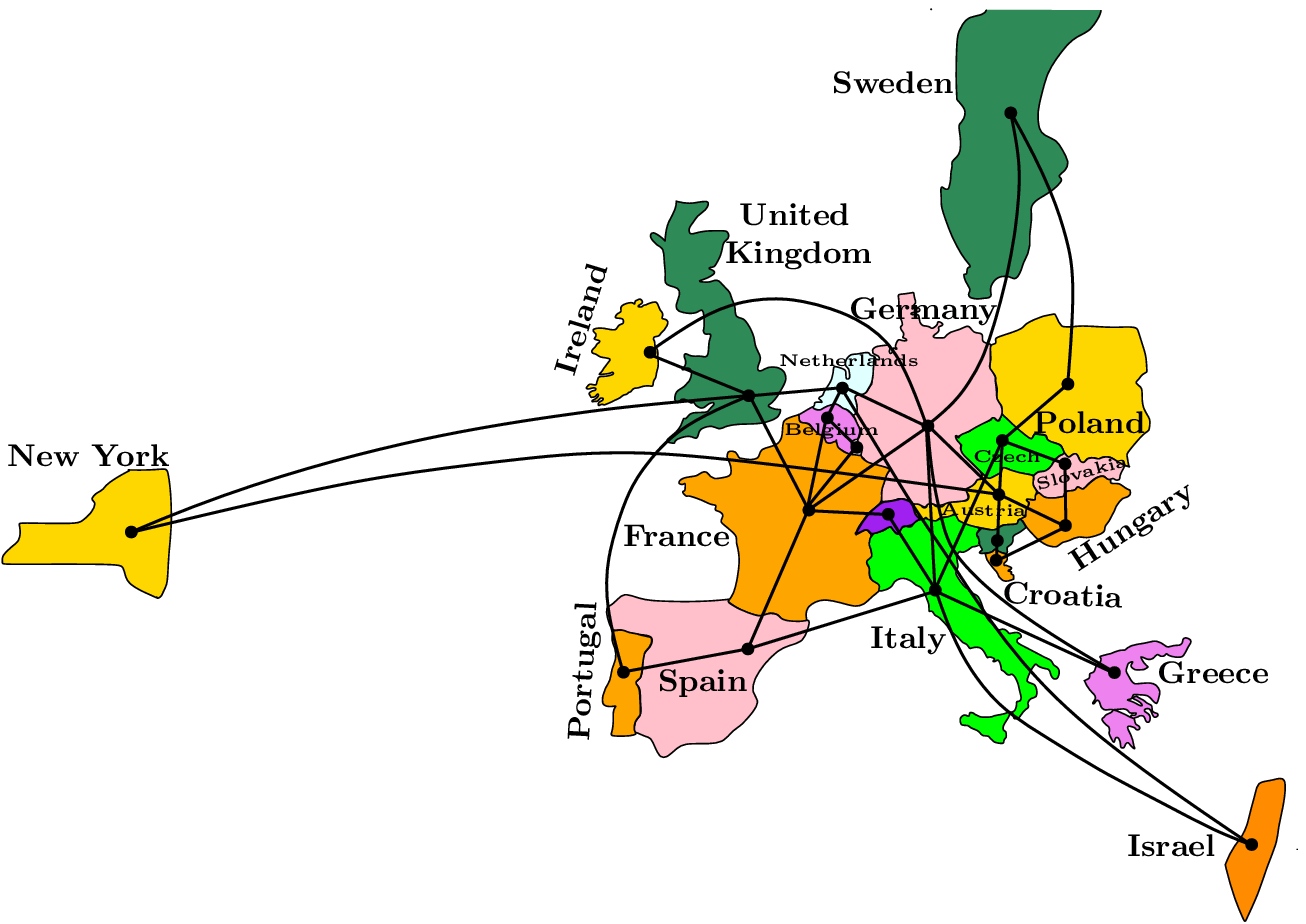}
\caption{Geant network with $22$ nodes and $36$ bidirectional links.}\label{fig:GeantNetwork}
\end{minipage}
\end{figure}

To gain a comprehensive performance view, five different traffic
scenarios with $5$, $10$, $20$, $50$, and $100$ flows are
considered. Note that a flow is characterized by its deadline, in
addition to origin and destination nodes. Hence for the small network,
when there are many flows, some will have the same origin and
destination, though different deadlines.  The origin and destination of
each flow are uniformly and randomly chosen from the network
nodes. The sizes of flows are also uniformly generated in the range of
$[1,100]$. To systematically study the impact of deadline, we set the
deadline of flow $f$ to $d_f=\alpha e_f$, where $\alpha$ is referred
to as deadline factor and $e_f$ denotes the earliest possible
completion time of flow $f$ assuming that all network capacity is
available to this flow.

We set $\alpha$ to obtain two types of scenarios, tight and
moderate. In the tight-deadline scenario, $\alpha$ is set such that
the instances are close to the feasibility/infeasibility
boundary. Next, we increase $\alpha$ by $30\%$, giving the
moderate-deadline scenario. In TSA, the number of time slices is a
multiple ($1$x, $2$x, $3$x) of the number of flows. In 1x, the time
slices are defined by the deadlines, the first slice is from time zero
to the first deadline, the second slice is from the first deadline to the
second deadline, and so on. In $2$x and $3$x, the time slices are
created, starting from the 1x case, followed by subdividing the
longest time slice into two equal ones repeatedly, until the desired
number of slices is obtained.

For each network topology and traffic scenario, we generate $10$
instances and report the average performance.  For Small and Softlayer
networks, we use a computational time limit of $500$ seconds for cases
of up to $50$ flows, and a time limit of $1000$ seconds when there are $100$
flows. For the Geant network, the corresponding time limits are $1500$
and $3000$ seconds.  The experiments were run on a Core i7 PC with a
CPU 2901 MHz and 16 Gigabyte RAM, running the operating system Windows
10. The Gurobi optimizer
\cite{Gurobi} is used for solving the mathematical models in TSA, CGA, and MFA.

\subsection{Performance Results for Tight Deadlines}

With tight deadlines, CGA is always able to obtain the optimal
solution of all instances, whereas TSA and MFA fail to deliver a
feasible solution for many and sometimes all of the
instances. Therefore, we first report the failure rates of TSA and
MFA, as shown in Table~\ref{table:tightdeadlines}.  The failure rate
is defined as ratio of the number of instances for which a solution could
not be obtained over the total number of instances. For TSA, the
infeasibility can be of two reasons. First, the instance with time
slicing is in fact infeasible due to an insufficient number of time
slices. Second, the instance could be feasible but the solver is not
able to verify it within the time limit. For MFA, infeasibility is
purely due to the algorithm, not because of the time limit.

\begin{table}[ht!]
\small
\centering
\begin{tabular}{ c c c c c c }
\hline \hline
&& & TSA (in~\%)&&\\
\cline{3-5}
Network& $F$&$1$x& $2$x& $3$x& MFA (in~\%)\\
\hline \hline
 &5 &80\textbar0 & 80\textbar10 &60\textbar30 &90\\

& 10 & 40\textbar50 &30\textbar40&30\textbar40 &100\\

Small& 20 & 30\textbar70&10\textbar90&0\textbar80 &100\\

& 50 & 40\textbar60&40\textbar60&0\textbar100 &100\\

& 100 & 0\textbar100&0\textbar100&0\textbar100 &100\\
\hline
\hline
 &5 &100\textbar0 & 70\textbar20 &70\textbar20 &100\\
& 10 & 80\textbar20 &50\textbar40&30\textbar60 &100\\

Softlayer& 20 & 10\textbar90&10\textbar80&0\textbar80 &100\\

& 50 & 30\textbar70&10\textbar90&0\textbar100 &100\\

& 100 & 0\textbar100&0\textbar100&0\textbar100 &100\\
\hline
\hline
 &5 &40\textbar30 & 30\textbar20 &30\textbar20 &40\\

& 10 & 30\textbar40 &40\textbar10&30\textbar20 &90\\

Geant& 20 & 10\textbar70&0\textbar60&0\textbar40 &90\\

& 50 & 0\textbar100&0\textbar100&0\textbar100 &100\\

& 100 & 0\textbar100&0\textbar100&0\textbar100 &100\\
\hline
\end{tabular}
\caption{Failure rates for TSA and MFA for instances with tight deadlines.
All values represent the percentage. For TSA, with notation p\textbar
q, p stands for the percentage of instances that are infeasible and q
represents the the percentage that a solution is not found due to time
limit. CGA has zero failure rate and therefore it is not included in
the table.}
\label{table:tightdeadlines}
\end{table}

On average, for Small, Softlayer, and Geant networks, TSA fails for
$90\%$, $95\%$, and $74\%$ of the instances respectively. For Small
network, this is mainly because the number of time slices is too
few in relation to the tight deadlines. This is also the case of
Softlayer and Geant networks, where there are relatively small numbers
of flows. When there are many flows, the failure is to a large extent
caused by the time limit.  By increasing the number of time slices, fewer
instances tend to be infeasible in TSA. This is expected by the
approach of time slicing. On the other hand, using a greater number of
time slices makes the problem size significantly larger, and consequently
the solution time becomes the bottleneck.  Note that TSA has a lower
failure rate for $5$, $10$, and $20$ flows in Geant Network. The reason is
that Geant is larger than the other two networks, hence it is easier
to obtain feasibility when there are few flows.

MFA has an extremely high failure rate. Feasibility is
attained only for Geant network and some of the instances with small
numbers of flows. Thus, an intuitive heuristic, such as MFA, that uses
deadline as priority is not a good choice if the deadlines are
tight.

\begin{figure}[ht!]
\vspace{-2mm}
\includegraphics[width=0.45\textwidth]{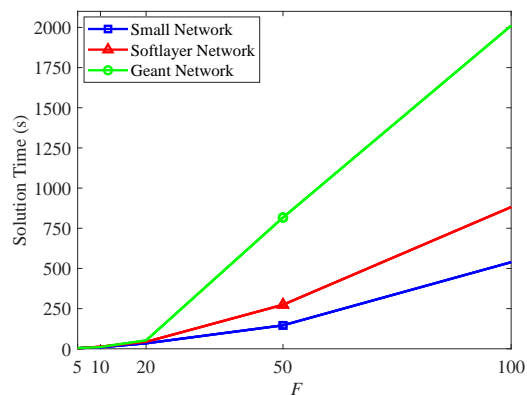}
  \vspace{-2mm}
   \caption{Solution time of CGA for the tight-deadline instances.}
  \vspace{-1mm}
  \label{fig:soltime_tight}
\end{figure}

The average time used by CGA to solve the instances is shown in
Figure~\ref{fig:soltime_tight}. CGA never hits the time limit, hence
the solutions are optimal. The time increases with respect to network
size as well as the number of flows. However, the rate of increase is
relatively moderate. For example, when the number of flows grows from
50 to 100, the increase in solution time is less than 3 times.

\subsection{Performance Results for Moderate Deadlines}

For instances with moderate deadlines, we examine three performance
aspects: 1) the optimality gap of TSA and MFA, 2) the failure rate of
these two approaches, and 3) the solution time of them in comparison
to CGA.  The numerical results are shown in
Figures~\ref{fig:smallPerform}-\ref{fig:geantPerform}.  We make the
following observations based on the results.

\subsubsection{Optimality gap and failure rate}

TSA and MFA show consistent results for the three networks.  As the
number of flows increases, the optimality gap of TSA decreases. This
is because the base number of time slices equals $F$, hence the
granularity increases with $F$.  For example for Softlayer network,
the gap by TSA($1$x) is about $12\%$ for $5$ flows and decreases to
only $3\%$ for $50$ flows. It can be seen that the optimality gap for
$100$ flows increases again, this is due to the fact that TSA hits the
time limit, see Figure~\ref{fig:softPerform}(a).  The failure rate
increases very significantly, see Figure~\ref{fig:softPerform}(b),
because of a higher risk that TSA terminates pre-maturely due to the
time limit. Moreover, the effect of scaling up, i.e., going from $1$x
to $2$x and $3$x, is apparent for small $F$, because the granularity
becomes significantly improved.  For large $F$, it has little impact
on optimality gap, but leads often to infeasibility as the problem size grows
considerably. For Softlayer network and $100$ flows, for example, the
failure rate of TSA($1$x) is $80\%$, and increases to $90\%$ for both
TSA($2$x) and TSA($3$x).

The optimality gap of MFA is clearly larger than that of TSA, and
peaks at approximately 20\% for Small network.  Note that the
optimality gap of MFA grows first, but then decreases in $F$. One
explanation is that the sub-optimality of MFA, due to imposing
priority strictly following deadlines, is first magnified by the
problem size, here the number of flows. However, when there are many
flows, the scheduling order in high-quality solutions become more
coherent with the deadlines, as it is less likely that combining flows
with deadlines being far apart will result in feasibility. As for the
failure rate, MFA performs well in this aspect for the three networks.

\subsubsection{Solution time}

The solution times are the average values over the instances for which
TSA and MFA were able to obtain a feasible solution within the time
limit. As a general trend, TSA is faster in delivering its solution
for small number of flows, see
Figures~\ref{fig:smallPerform}(d)-\ref{fig:geantPerform}(d), and in
such cases feasibility is not an issue for TSA.  Note that the
difference in time can be very significant.  For Geant network and 20
flows, for example, CGA needs almost 50~seconds, whereas the solution
time of TSA(1x) is only a couple of seconds. The solution of the
latter is not optimal, however the gap is quite small.  For large
number of flows, however, CGA clearly outperforms TSA in time.  For
MFA, the solution time is very attractive -- it is in order of one or even two magnitudes faster than TSA and CGA. However MFA has the
highest optimality gap.

As the overall observation, CGA outperforms TSA and MFA, in delivering
optimum with zero failure rate. On the other hand, even if neither TSA nor
MFA gives satisfactory results by themselves, they can be used
together with CGA, either to speed up the latter, or to reduce the
overall time when some tolerance of optimality is accepted. These
aspects are examined below.

%-----------------------------small network
\begin{figure}[!h]
\vspace{-5mm}
  \begin{subfigure}[b]{0.43\textwidth}
    \includegraphics[width=\textwidth]{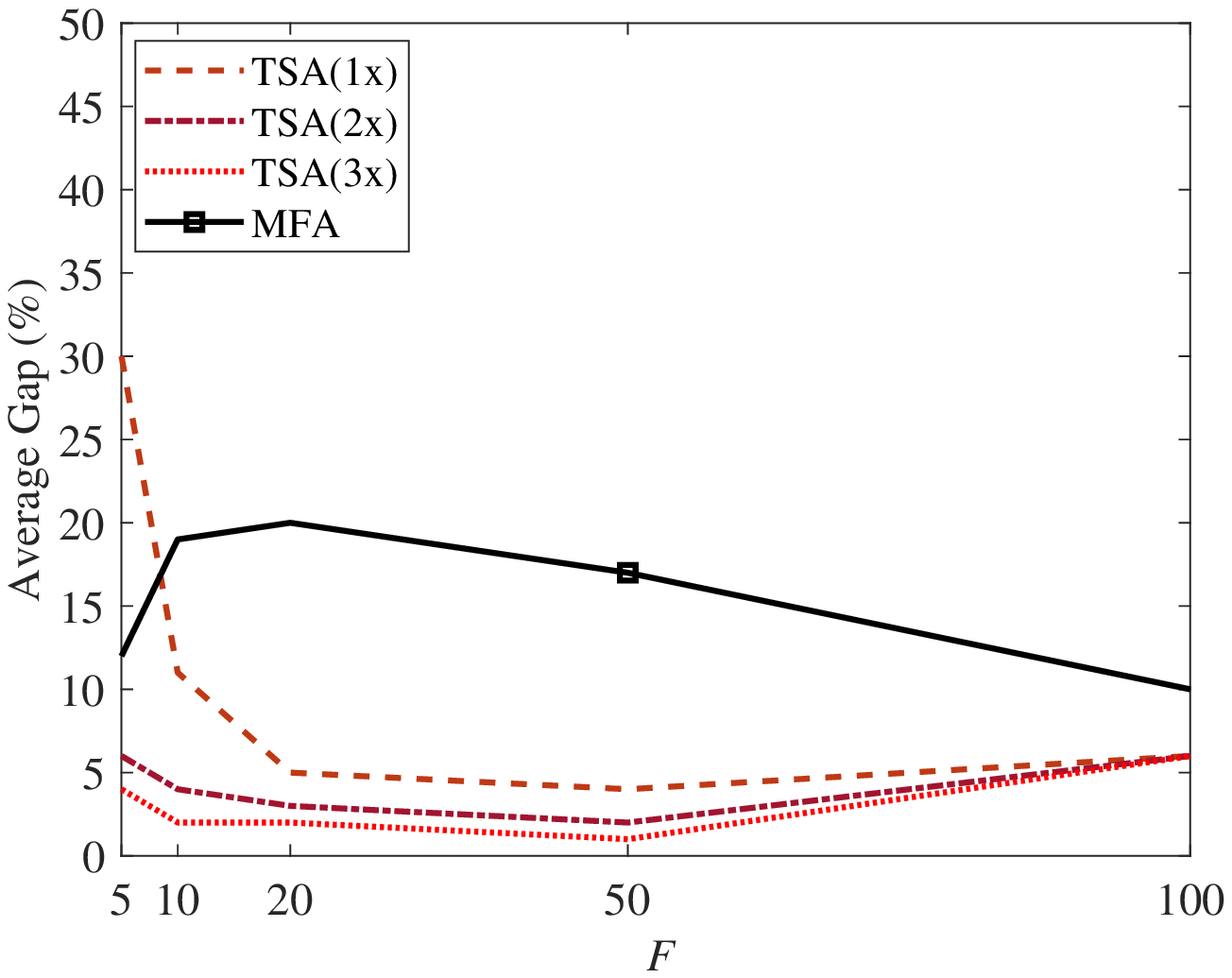}
    \vspace{-5mm}
    \caption{Optimality gap of TSA and MFA.}
    \label{fig:smallgap}
  \end{subfigure}
\begin{subfigure}[b]{0.43\textwidth}
\scriptsize
\centering
\begin{tabular}{c cccc }
\hline\hline
& & TSA (in~\%)&&\\
\cline{2-4}
$F$&$1$x& $2$x& $3$x& MFA (in~\%)\\
\hline\hline
 5  &    0\textbar0     &    0\textbar0    & 0\textbar0    &0\\

10  &    0\textbar0     &    0\textbar0     & 0\textbar0    &0\\

20  &   0\textbar0    &      0\textbar0     & 0\textbar0    &0\\

50  &   0\textbar0    &      0\textbar0     &0\textbar0    &0\\

100 &   0\textbar0    &   0\textbar10    &0\textbar20   &0\\
\hline
\end{tabular}
\vspace{5mm}
    \caption{Failure rate of TSA and MFA.}
    \label{fig:small_fail}
\end{subfigure}
%\end{figure}
%\begin{figure}[H]%\ContinuedFloat
%\vspace{0mm}
\begin{subfigure}[b]{0.43\textwidth}
    \includegraphics[width=\textwidth]{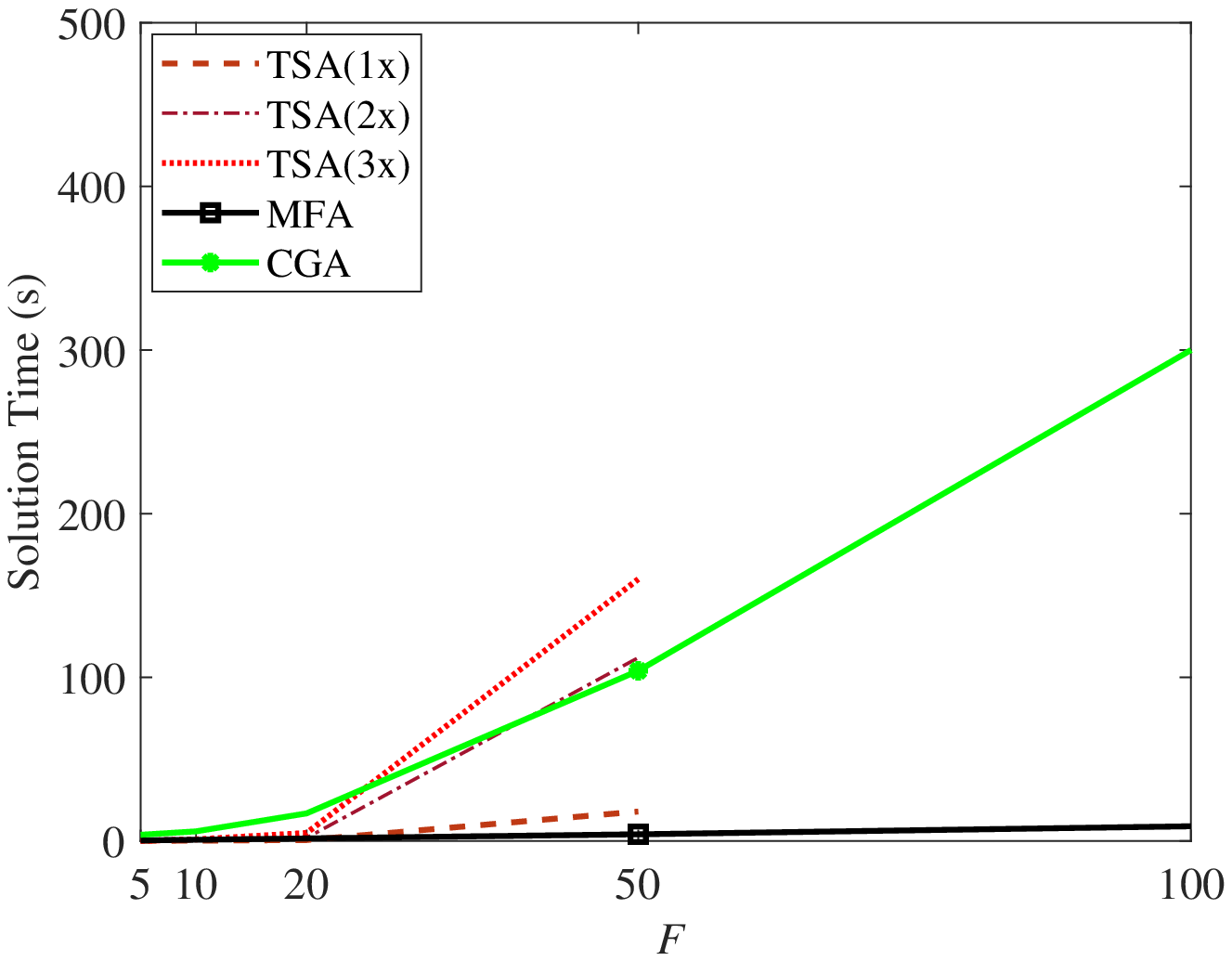}
    \vspace{-3mm}
    \caption{Solution time over all number of flows.}
    \label{fig:small_solTime}
\end{subfigure}
\begin{subfigure}[b]{0.43\textwidth}
    \includegraphics[width=\textwidth]{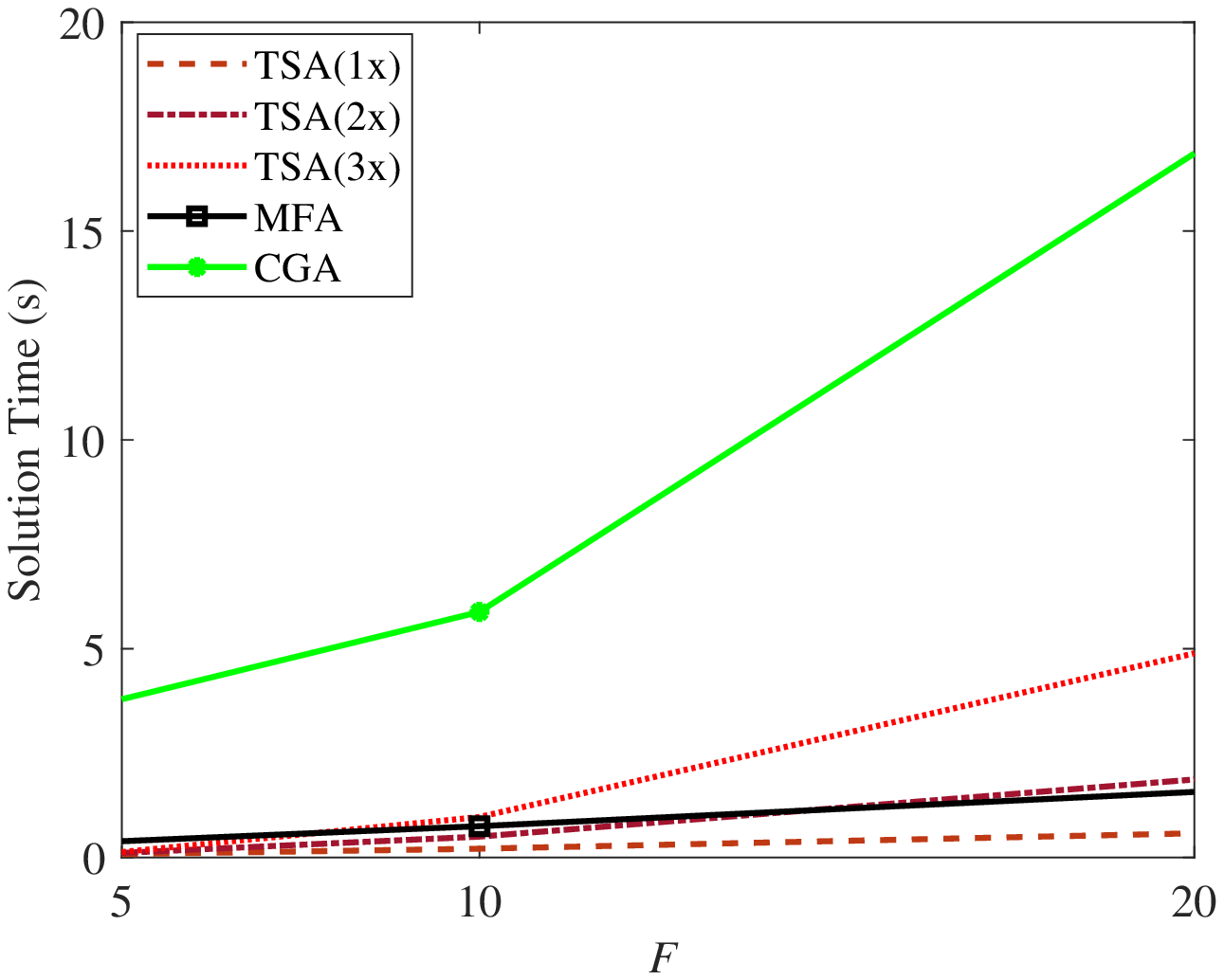}
    \vspace{0mm}
    \caption{Solution time over small number of flows.}
    \label{fig:small_smallSolTime}
\end{subfigure}
  \caption{Performance results for Small network.}
  \vspace{0mm}
  \label{fig:smallPerform}
\end{figure}
\vspace{-3mm}

%-------------------------------softlayer
\begin{figure}[H]
  \begin{subfigure}[b]{0.43\textwidth}
    \includegraphics[width=\textwidth]{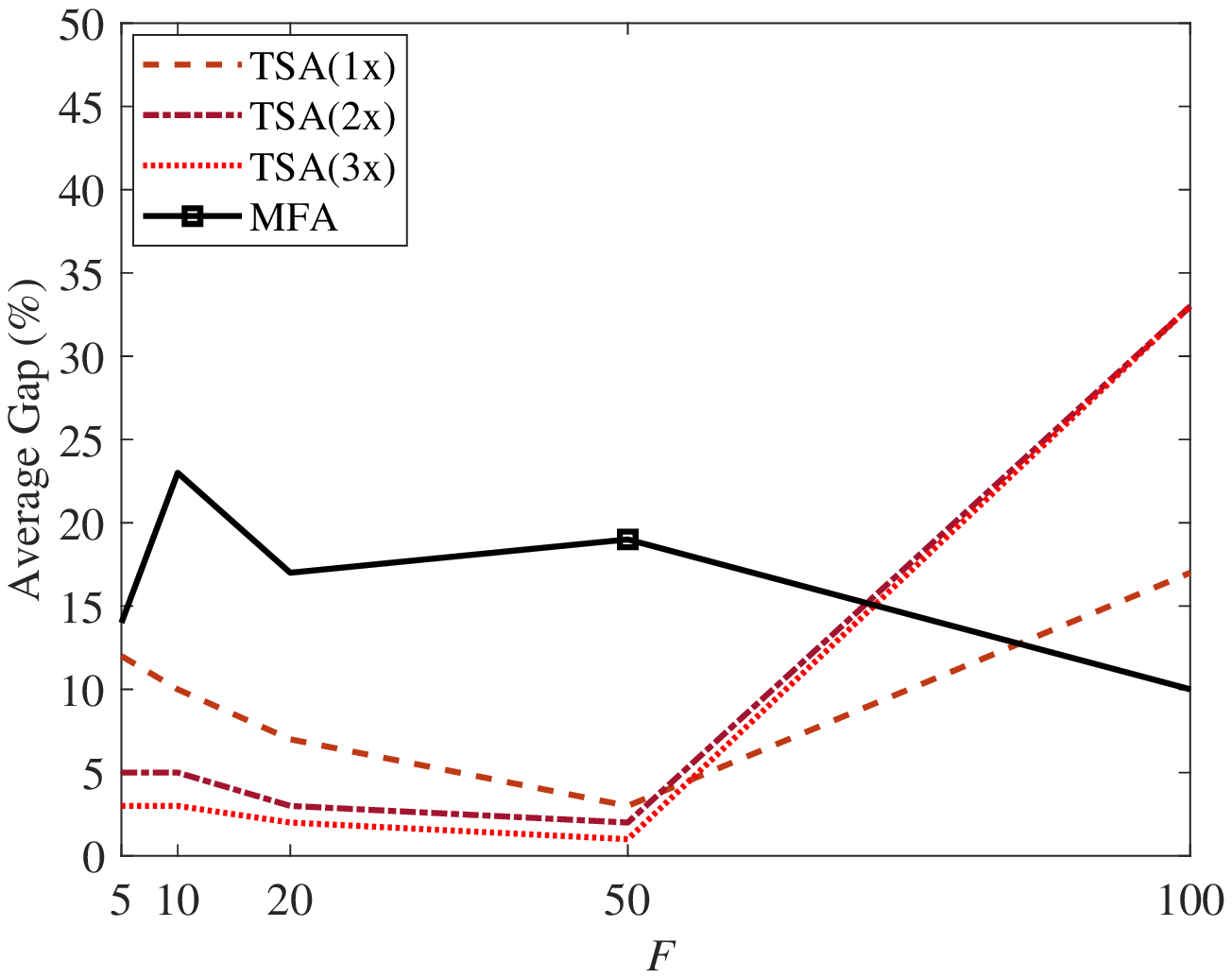}
    \vspace{-5mm}
    \caption{Optimality gap of TSA and MFA.}
    \label{fig:softgap}
  \end{subfigure}
\begin{subfigure}[b]{0.43\textwidth}
\scriptsize
\centering
\begin{tabular}{c ccc c }
\hline\hline
& & TSA (in \%)&&\\
\cline{2-4}
$F$&$1$x& $2$x& $3$x& MFA (in \%)\\
\hline\hline
 5  &    0\textbar0     &    0\textbar0    & 0\textbar0    &0\\

10  &    0\textbar0     &    0\textbar0     & 0\textbar0    &0\\

20  &   0\textbar0    &      0\textbar0     & 0\textbar0    &0\\

50  &   0\textbar0    &      0\textbar0     &0\textbar30    &0\\

100 &   0\textbar80    &   0\textbar90    &0\textbar90   &0\\
\hline
\end{tabular}
\vspace{7mm}
    \caption{Failure rate of TSA and MFA.}
    \label{fig:soft_fail}
\end{subfigure}
%\end{figure}
%\begin{figure}[H]%\ContinuedFloat
%\vspace{-12mm}
    \begin{subfigure}[b]{0.43\textwidth}
    \includegraphics[width=\textwidth]{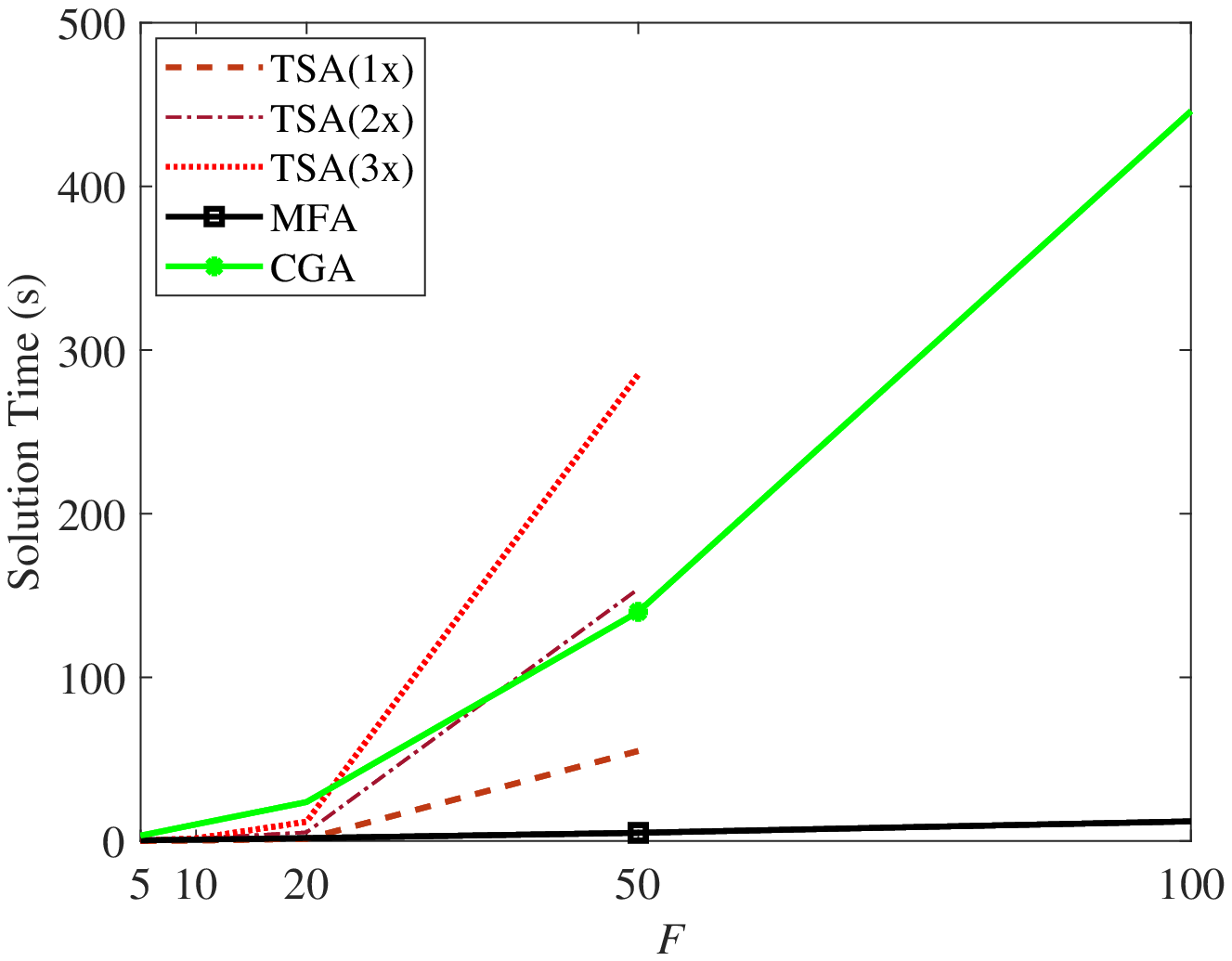}
    \caption{Solution time over all number of flows.}
    \label{fig:soft_solTime}
  \end{subfigure}
      \begin{subfigure}[b]{0.43\textwidth}
    \includegraphics[width=\textwidth]{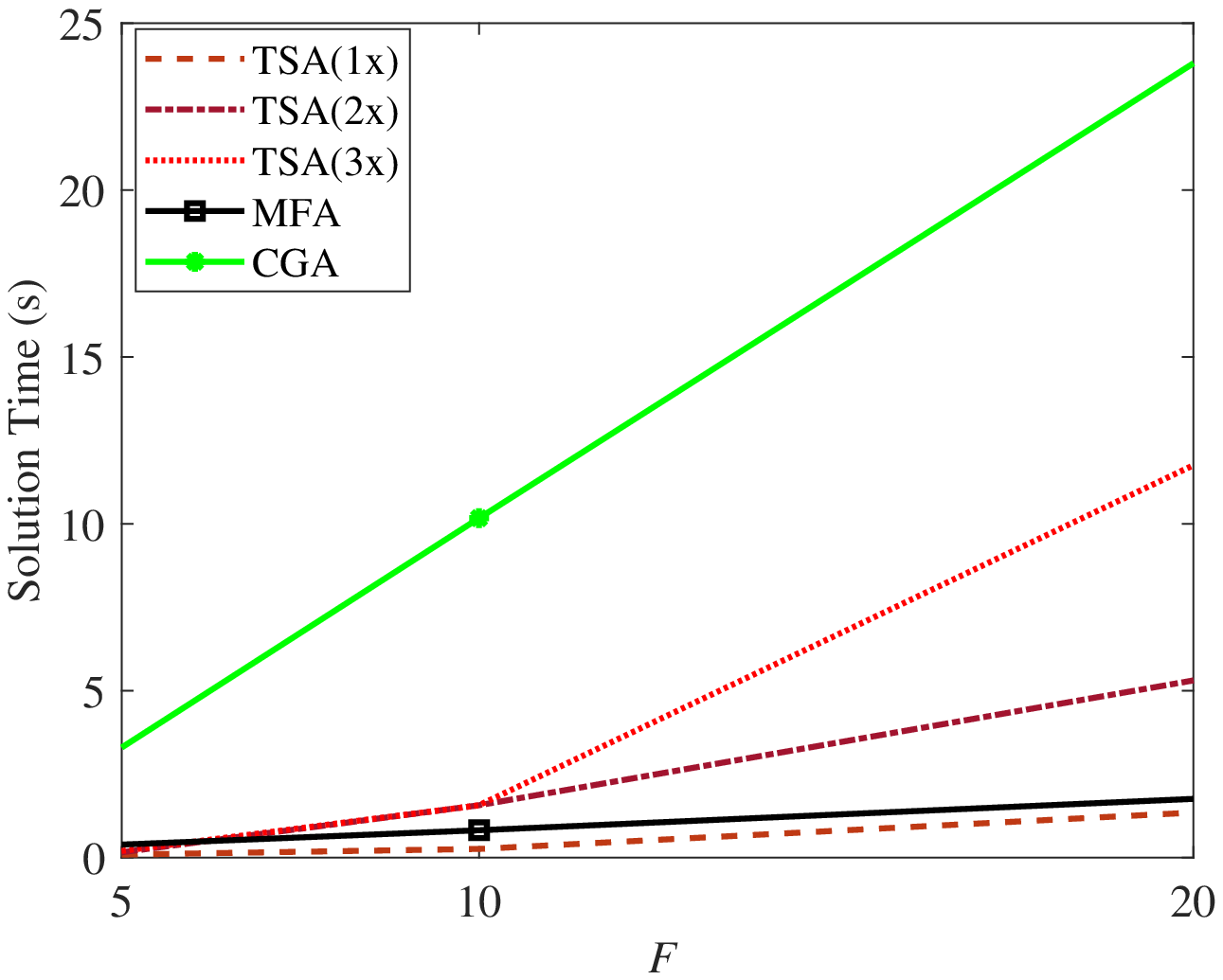}
    \caption{Solution time over small number of flows.}
    \label{fig:soft_smallSoltime}
  \end{subfigure}
  \caption{Performance results for Softlayer network.}
  \vspace{-2mm}
  \label{fig:softPerform}
\end{figure}

%-------------------------------Geant
\begin{figure}[H]
  \begin{subfigure}[b]{0.43\textwidth}
    \includegraphics[width=\textwidth]{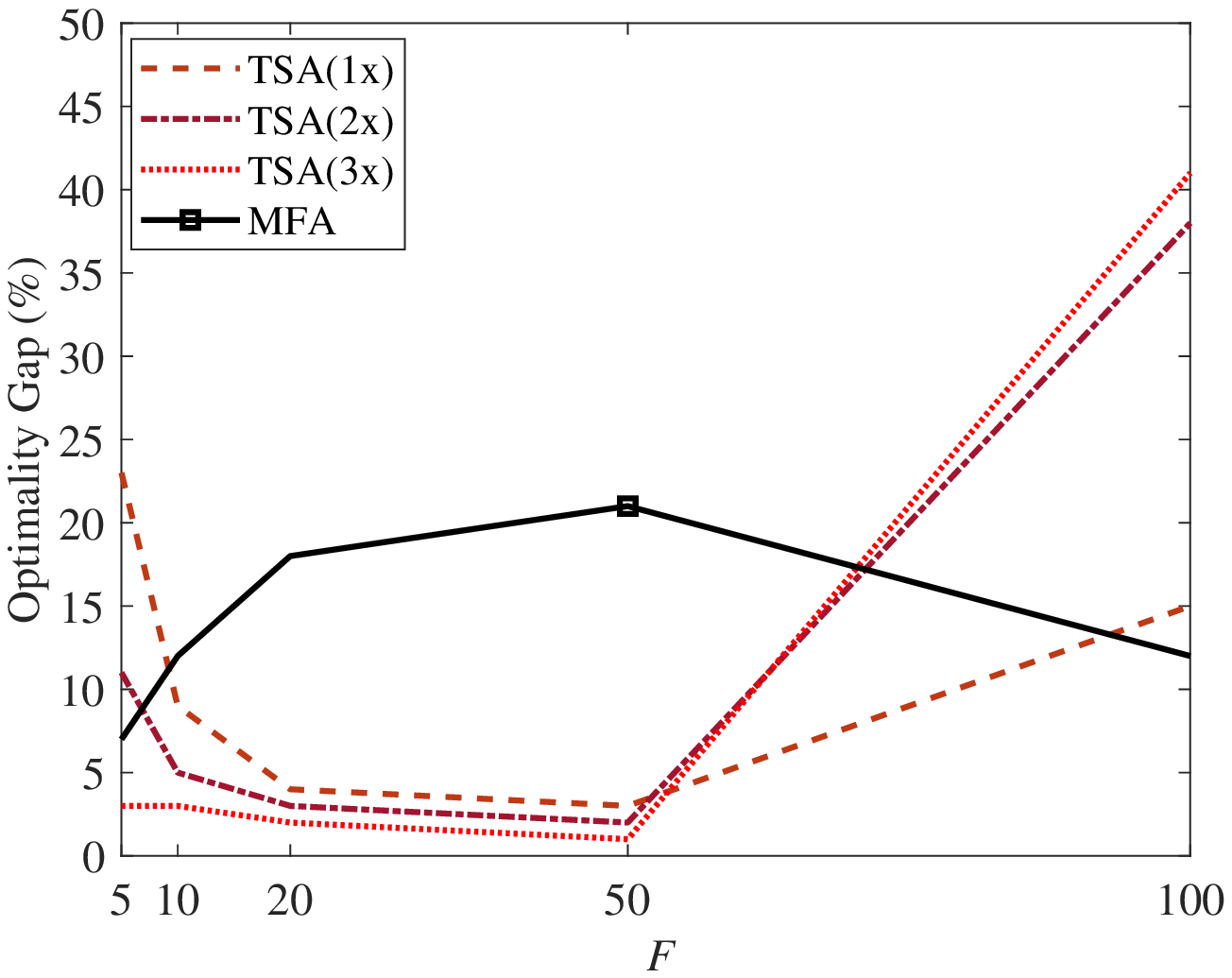}
    \caption{Optimality gap of TSA and MFA.}
    \label{fig:geantgap}
  \end{subfigure}
  \begin{subfigure}[b]{0.43\textwidth}
\scriptsize
\centering
\begin{tabular}{c ccc c }
\hline \hline
& & TSA~(in~\%)&&\\
\cline{2-4}
$F$&$1$x& $2$x& $3$x&MFA~(in~\%) \\
\hline \hline
 5  &    0\textbar0     &    0\textbar0    & 0\textbar0    &0\\

10  &    0\textbar0     &    0\textbar0     & 0\textbar0    &0\\

20  &   0\textbar0    &      0\textbar0     & 0\textbar0    &10\\

50  &   0\textbar0    &      0\textbar20     &0\textbar20    &10\\

100 &   0\textbar30    &   0\textbar80    &0\textbar90   &10\\
\hline
\end{tabular}
\vspace{7mm}
    \caption{Failure rate of TSA and MFA.}
    \label{fig:geant_fail}
\end{subfigure}
%\end{figure}
%\vspace{-12mm}
%\begin{figure}[H]%\ContinuedFloat
 \begin{subfigure}[b]{0.43\textwidth}
    \includegraphics[width=\textwidth]{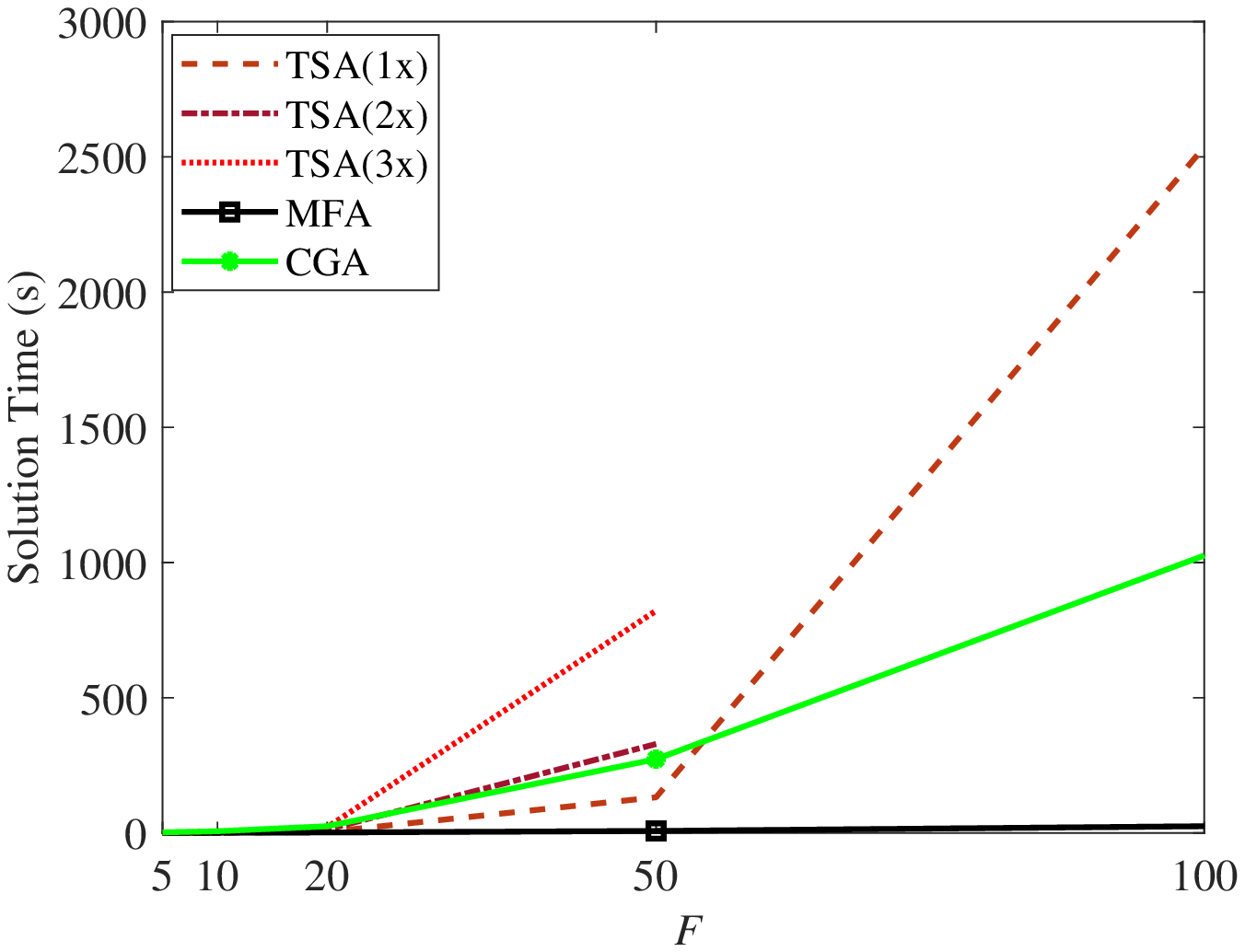}
    \caption{Solution time over all number of flows.}
    \label{fig:geant_solTime}
  \end{subfigure}
 \begin{subfigure}[b]{0.43\textwidth}
    \includegraphics[width=\textwidth]{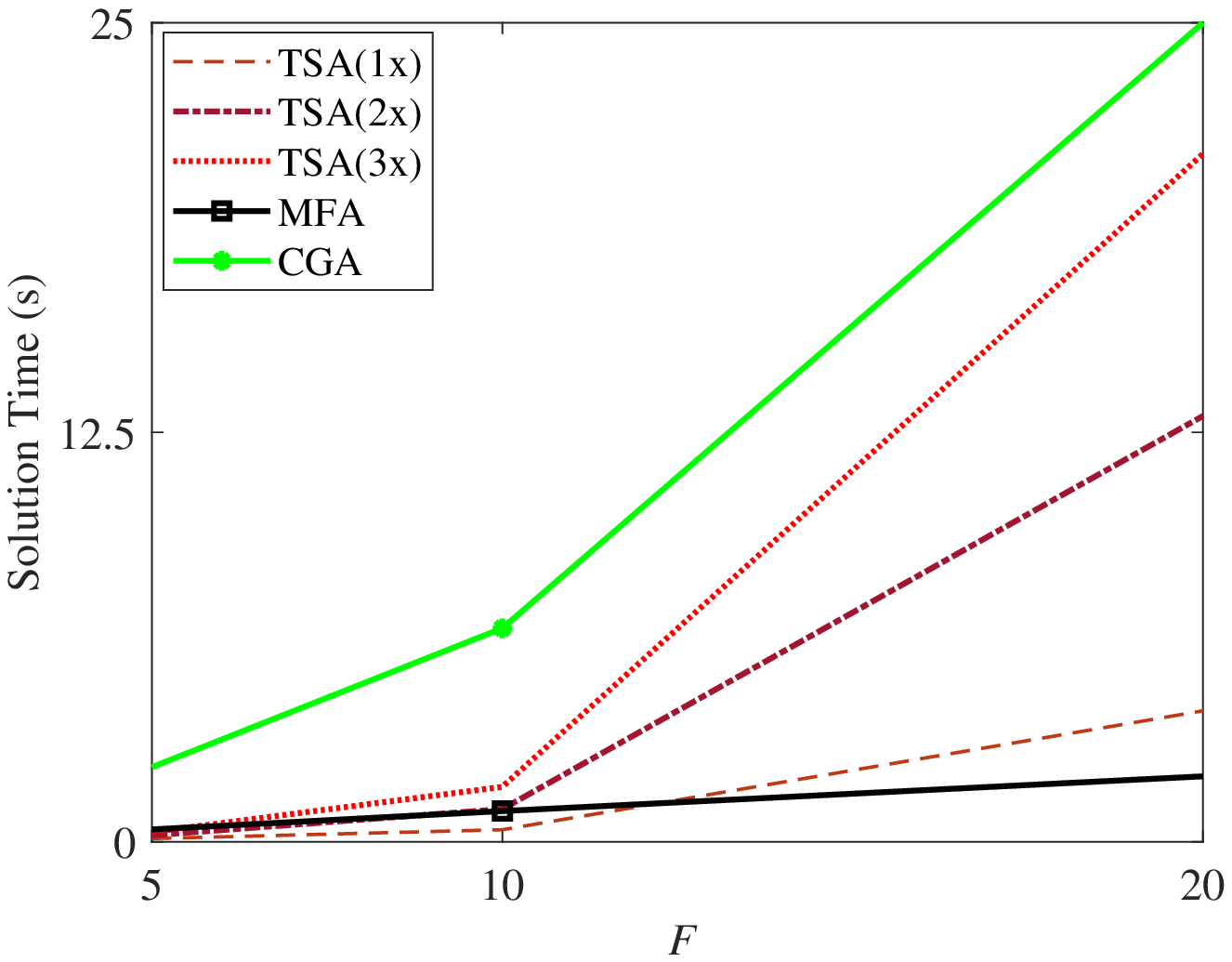}
    \caption{Solution time over small number of flows.}
    \label{fig:geant_smallSolTime}
  \end{subfigure}
   \caption{Performance results for Geant network.}
   \vspace{-2mm}
  \label{fig:geantPerform}
\end{figure}
\vspace{-1mm}

%-----------------------------------------------------------soltime

\subsection{Combining CGA with MFA and TSA}

In this section, we consider using CGA in conjunction with MFA and/or
TSA. Recall that, to reach the initial feasible solution (IFS), CGA runs
Phase~I in which the objective is to minimize a penalty function
of infeasibility.  We refer to the process of reaching optimality,
starting from the IFS, as Phase~II.

First, we observe that when MFA does deliver a feasible solution, the
time required is much shorter than that of Phase~I, and the solution
is generally of better quality than the IFS of CGA.
Hence, to speed up CGA, we can use the solution of MFA as an IFS
for CGA.  The time of Phase~I is reduced therefore to the running time
of MFA. We refer to this combination as MFA-CGA.  Second, as discussed
in Section~\ref{sec:timeSlicing}, the LP relaxation of TSA (rTSA) can
be used to provide a lower bound of the global optimum. Hence, if a
tolerance on the optimality gap is present, the bound from rTSA can be
used to terminate CGA as soon as the current solution of CGA meets the
tolerance parameter with respect to the lower bound. As a result, the
time necessary for Phase~II is reduced. This is particularly useful
for large-scale instances.  We refer to this combination as
rTSA-CGA($p$), where $p$ represents the percentage value of the
tolerance parameter. Finally, combining both MFA and rTSA with CGA
is denoted by MFA-rTSA-CGA($p$).

Table \ref{table:savingtime} summarizes our findings. We show the
average solution times of Phase~I and Phase~II of CGA,
and the time reductions achieved.
For MFA-CGA, the time reduction of Phase~I is the relative difference of the
computing time of Phase~I of CGA and the running time of MFA.  For
rTSA-CGA, the time reduction reported for Phase~II is the relative difference
of the computing time of Phase~II of CGA and the sum of the time for
solving rTSA and the time needed for Phase~II before satisfying the
10\% optimality tolerance. Note that MFA is not used for Phase~I in
the results of rTSA-CGA.

\begin{table}[ht!]
	\small
	\centering
\setlength\tabcolsep{2 pt}
	\begin{tabular}{c c c c c c c c}
		\hline
		\hline
		& &\multicolumn{2}{c}{Average Solution Time of CGA (s)}&& \multicolumn{3}{c}{Average Time Reduction (in \%)}\\
		\cline{3-4}
		\cline{6-8}
		Network & $~~~~F~~~~$ & Phase~I & Phase~II && MFA-CGA~ & rTSA-CGA($10\%$)~ & MFA-rTSA-CGA($10\%$) \\
		  & & & && Phase I& Phase II &  Total\\
		\hline
		\hline
		& 5   &0.95     &2.07     &&58    &66    &74 \\
		& 10  &2.32     &3.22     &&64    &48    &57 \\
		Small       & 20  &6.97     &9.21     &&80    &54    &76 \\
		& 50  &35.25    &60.98    &&88    &58    &88 \\
		& 100 &169.05   &126.42   &&95    &52    &96 \\
		\hline
		\hline
		& 5   &0.69     &0.97     &&51  &48 &52 \\
		& 10  &2.16     &6.07     &&65  &57 &67 \\
		Softlayer   & 20  &7.10     &14.10    &&78  &57 &81 \\
		& 50  &49.28    &84.90    &&90  &55 &86 \\
		& 100 &262.87   &187.18   &&95  &44 &95 \\
		\hline
		\hline
		& 5   &0.97     &0.74      &&65  &34 &75 \\
		& 10  &2.25     &3.25      &&69  &49 &75 \\
		Geant       & 20  &9.95     &13.60     &&82  &57 &81 \\
		& 50  &134.45   &136.05    &&93  &54 &86 \\
		& 100 &606.41   &377.93    &&95  &63 &94 \\
		\hline
	\end{tabular}
	\caption{Average solution time of CGA and average time reduction.}
	\label{table:savingtime}
\end{table}

We observe that both phases of CGA require similar time (within order
of magnitude). Therefore time reduction is of importance for both.

MFA-CGA yields for very significant time reduction of
Phase~I, from $50\%$ with few flows up to $95\%$ with 100 flows.  The
time reduction increases with respect to the number of flows. The
reason is that the solution time of MFA grows with a much slower rate
than that of running Phase~I of CGA in its original design.  We
remark that the solution of MFA is typically better than the IFS of
the original CGA. Hence, even though not shown in the table,
MFA-CGA also gives some time reduction in Phase~II.

Having an optimality tolerance of $10\%$ and using the bound of rTSA,
i.e., rTSA-CGA($10\%$), leads to significantly less time in Phase~II.
The time reduction ranges between $34\%$ and $66\%$.  We observe no
correlation of the time reduction with the network size or the number of
flows. For a majority of cases, the percentage values are lower than those
for Phase~I achieved via MFA-CGA.  This is however counter-balanced,
to some extent, by the fact that Phase~II sometimes takes more time
than Phase~I in CGA, and in such a case the reduction of the former has more impact.

The combination MFA-rTSA-CGA($10\%$) gives a substantial reduction of
the total solution time, starting from $50\%$ for $5$ flows and
reaching $96\%$ for $100$ flows. Hence the approach is useful for
dealing with large-scale scenarios where the problem size is mainly due to
the number of flows.

\begin{figure}[!h]
   \begin{subfigure}[b]{0.43\textwidth}
    \includegraphics[width=\textwidth]{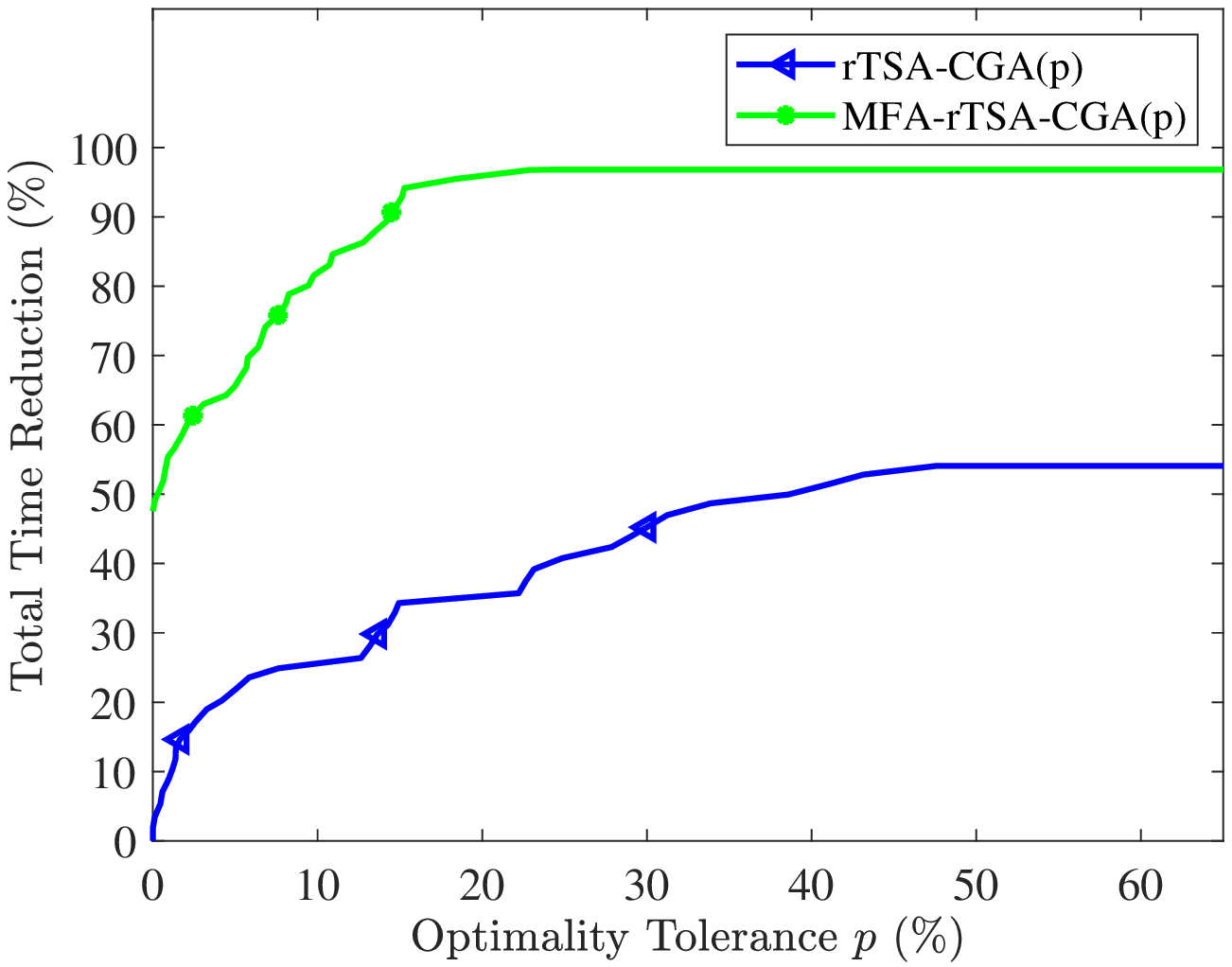}
    \caption{Softlayer network with 50 flows.}
  \end{subfigure}
     \begin{subfigure}[b]{0.43\textwidth}
    \includegraphics[width=\textwidth]{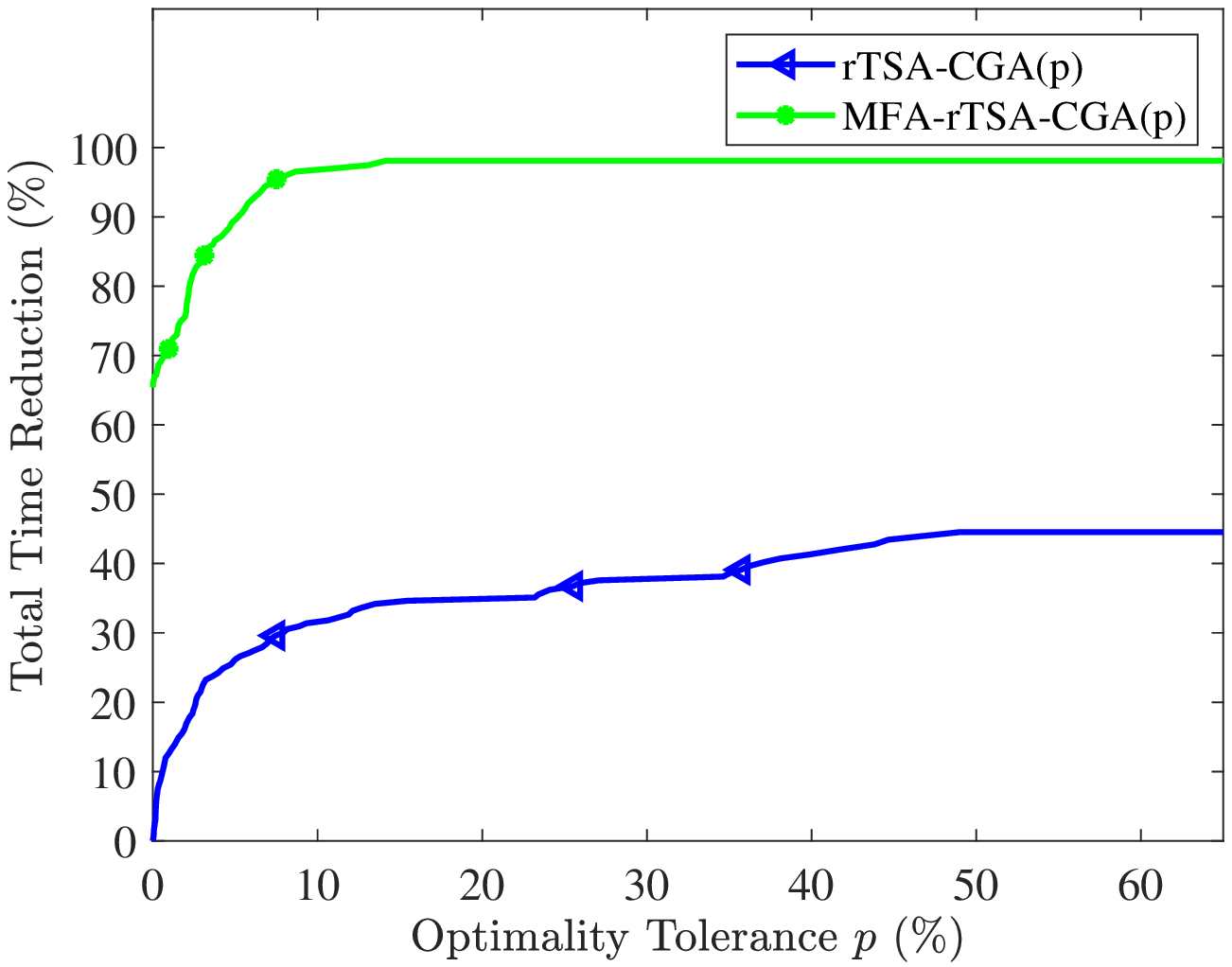}
    \caption{Geant network with 100 flows.}
  \end{subfigure}
   \caption{Times reduction by rTSA-CGA($p$) and MFA-rTSA-CGA($p$)
with respect to $p$. The time reduction is calculated in relation to the
overall time of CGA.}
   \vspace{0mm}
  \label{fig:savingtimes}
\end{figure}

Figure~\ref{fig:savingtimes} illustrates how the time reduction by
rTSA-CGA($p$) and MFA-rTSA-CGA($p$) varies as a function of optimality
tolerance parameter $p$ for two representative instances. Obviously,
higher tolerance of optimality gap means larger time reduction. This
growth stops at the point where the optimality gaps of the solution
provided by MFA and the (original) IFS of CGA are within the
tolerance, for MFA+rTSA+CGA($p$) and rTSA+CGA($p$), respectively.
As expected, MFA-rTSA-CGA($p$) offers clearly larger reduction of
the overall time than rTSA-CGA($p$) for all $p$-values.  This is
because the latter focuses on Phase~II only.  In MFA-rTSA-CGA($p$), both
phases are shortened in time, and the reduction on Phase~II is a joint
effect of rTSA via bounding and of MFA that provides a better IFS than
what CGA had. Using MFA-rTSA-CGA($p$), the CGA time is reduced by
one order of magnitude for $p=10\%$ and $p=15\%$ for the two scenarios.

\section{Conclusions}
\label{sec:concolusion}

We have investigated flow routing and scheduling with the presence of
deadline constraints and capacity allocation with discrete units.  In
addition to examining problem complexity, we have considered three
solution approaches, TSA, CGA, and MFA. They use time slicing, problem
reformulation and column generation, and multicommodity maximum flow,
respectively. Among them, CGA delivers global optimum.  Our
performance evaluation shows that, with tight deadlines, TSA and MFA
often fail to obtain a feasible solution. When the deadlines are less
stringent, TSA provides solution faster than CGA if the number of
flows is relatively small, though the optimality gap is
non-negligible.  MFA is fastest among the three approaches when
infeasibility is not an issue. However, it also has the largest
optimality gap. Overall, CGA represents a viable solution approach for
global optimality. Moreover, there are relevant use cases of the
solution from MFA and the bound from TSA, in the context of CGA,
particularly for large-scale instances.

Further work includes extensions of the problem, to which adaptions of
the proposed approaches will be studied. One specific case is the data backup
problem studied in \cite{Yao2015} that applies TSA for problem
solution.  In this problem, a source has multiple candidate
destinations for data backup. At any time, a source may choose at most
one destination, and a destination may be used by at most one
source. No deadline is present.  We remark that CGA can be adapted by
setting a large value for all flow deadlines, and tailoring the
subproblem formulation to no common source or destination is used by
the flows with positive rates in the rate vectors. Some preliminary
results indicate the relative performance between TSA and CGA is
coherent with those in Section~\ref{sec:performance_evaluation} in
terms of solution time and optimality gap. Extensive performance
evaluation as well as extensions to other related problems
are subject to further study.

\begin{appendices}

\section{}
\label{sec:notation}

\begin{table}[H]
\renewcommand{\arraystretch}{1.1}
\begin{tabular}{clcl}
\hline\hline
Symbol&Definition&Symbol&Definition\\
\hline \hline
 $\mathcal{G}$                           &  network &
 $\mathcal{N}$                          & set of nodes                        \\
 $N$                                     &   number of nodes&
 $\mathcal{A}$                              &   set of arcs
\\
 $A$                                      &   number of arcs&
 $(i,j)$                                      &  arc from $i$ to $j$
\\
$c_{ij}$ &    capacity of arc $(i,j)$   &
 $\mathcal{U}$                              &  set of discrete capacity units
\\
 $u_i$                                      &  size of capacity unit $i$ & $\mathcal{F}$                                &  set of flows
\\
  $F$                             &   number of flows&
  $o_f$                                      &   origin of flow $f$  \\
  $d_f$ & destination of flow $f$&
 $t_f$                                      &   deadline of flow  $f$
\\
$s_f$&size of flow $f$
\\  \hline
\end{tabular}
\caption{Basic notation.}
\end{table}

\vspace{-5mm}
\begin{table}[H]
\renewcommand{\arraystretch}{1.1}
\centering\singlespacing
\begin{tabular}{c p{15.3cm}}

\hline\hline
Symbol&Definition\\
\hline\hline
 $\mathcal{T}$                       &  set of time slices \\
 $\tau$                              &  a time slice   \\
 $|\tau|$                                 & length of time slice $\tau$\\
 $y_{fij}^\tau$                      & continuous variable, denoting the
rate of flow $f$ on arc $(i,j)$ in time slice $\tau$
\\
 $r_f^\tau$                          & continuous variable, indicating the
end-to-end rate of flow $f$ in time slice $\tau$ \\
 $w_\tau$                           & binary variable that takes value
one if any flow is scheduled in time slice $\tau$, and zero otherwise
\\
$z_{fij}^{m,\tau}$ &  integer variable that denotes
the number of times that flow $f$ uses capacity unit $u_m$ on arc $(i,j)$ in time slice $\tau$
\\ \hline
\end{tabular}
\caption{Notation related to the time slicing approach (TSA).}
\end{table}

\begin{table}[H]
\renewcommand{\arraystretch}{1.1}
\begin{tabular}{cl }
\hline \hline
Symbol&Definition\\
\hline \hline
 $\chi$                     & overall completion time
\\ \
 ${\bf{r}^*}$               & optimum rate vector                        \\
 $\delta^*$                 & the minimum time required to complete one of the flows
\\
 $f^*$                     & the flow that completes first among the flows
\\ \hline
\end{tabular}
\caption{Notation related to the max-flow based algorithm (MFA).}
\end{table}

\begin{table}[H]
\renewcommand{\arraystretch}{1.1}
\begin{tabular}{c p{15.3cm}}
\hline \hline
Symbol&Definition\\
\hline \hline
$\mathcal{V}$   &  set of rate vectors\\
 $v$             & an end-to-end rate vector in form of $[r_1^v, r_2^v,\dots,r_F^v]^T$ \\
 $r_f^v$         & the rate of flow $f$ in vector $v$ \\
 $\mathcal{V}_f$ &  set of rate vectors in which flow $f$ has positive end-to-end rate: $\{v \in \mathcal{V}|r_f^v>0\}$\\
 $y_{fij}$       &continuous variable, indicating the  rate of flow $f$ on arc (i,j) \\
 $r_f$           &continuous variable, denoting the rate of flow $f$ \\
 $x_v$           &continuous variable, denoting the time duration that vector $v$ is scheduled \\
 $z_{fij}^m$     &   integer variable, denoting the number of times that capacity
unit $u_m$ is used by flow $f$ on arc $(i,j)$\\
 $f^+$           & the first flow index with positive rate in a rate vector\\
 $\lambda_f^*$     &  optimal dual value corresponding to the demand constraint of flow $f$\\
 $\pi_f^*$         &optimal dual value corresponding to deadline constraint of flow $f$\\
$\mathcal{Q}_f$ &set of rate vectors in which $f$ is the first flow of positive rate:
$\{v|v \in \mathcal{V}_f \setminus (\mathcal{V}_0 \cup \dots\cup \mathcal{V}_{f-1})\}$
\\
 $p_f$ & a time point by which the data transmission of flow $f$ is completed    \\ \hline
\end{tabular}
\caption{Notation related to the column generation algorithm (CGA).}
\end{table}

\begin{table}[H]
\renewcommand{\arraystretch}{1.1}
%\begin{tabular}{c m{0.9\linewidth}}
\begin{tabular}{c p{15cm}}
\hline \hline
Symbol & Definition\\
\hline \hline
 $\alpha$        &  deadline factor
\\
 $e_f$          &  the earliest possible completion time of flow $f$ if all capacity of the network is given to this flow
  \\
 $1$x            &   use of $F$ time slices in TSA
\\
 $2$x            &   use of $2F$ time slices in TSA
\\
 $3$x            &   use of $3F$ time slices in TSA
\\ \hline
\end{tabular}
\caption{Notation related to the simulation setup.}
\end{table}

\section{}\label{NPhard1}

We adopt a polynomial-time reduction from the
3-satisfiability (3-SAT) problem that is NP-complete
\cite{garey1979computers}. Consider any 3-SAT instance with $m$
Boolean variables $n_1, n_2, \dots, n_m$, and $k$ clauses $c_1, c_2,
\dots, c_k$. A variable or its negation is called a literal. Denote by
$\hat{n}_i$ the negation of $n_i$, $i=1,2, \dots, m$. Each clause
consists of a disjunction of exactly three different literals, e.g.,
$n_1 \lor n_2 \lor \hat{n}_3$.  We use $\mathcal{Z}_i$ and
$\hat{\mathcal{Z}}_i$ to denote the sets of clauses in which variable
$n_i$ and its negation $\hat{n}_i$ appear, respectively. Also,
$\text{Z}_i$ and $\hat{\text{Z}}_i$ are used to denote their
respective cardinalities. We assume that no clause contains both a
variable and its negation, and any literal appears in at least one
clause and at most $k-1$ clauses. For any arc $(u,v)$ that we define
below, $u$ and $v$ are referred to as head and tail of the arc,
respectively. Similarly, the first and last node of a path are called
head and tail of the path.  We construct an IFDP instance as
follows. The number of nodes equals $N=\sum_{i=1}^{m}(\text{Z}_i
+\hat{\text{Z}}_i+2)+3k$.  The first $\sum_{i=1}^{m}(\text{Z}_i
+\hat{\text{Z}}_i+2)$ nodes are referred to as literal nodes, and the
last $3k$ nodes are referred to as clause nodes. The number of arcs is
$A=\sum_{i=1}^{m}(\text{Z}_i +\hat{\text{Z}}_i+2)+(6k+1)$; the nodes
indexed by the numbers in the two pairs of parentheses are referred to
as literal and clause arcs, respectively. The capacity of all arcs
equals one, which is also the unit in capacity allocation. The number
of flows is $F=m+k$, referred to as literal and clause flows,
respectively. The deadline of all flows is one time unit, i.e.,
$t_i=1$ for $i=1,\dots m+k$. The size of each flow is one, i.e.,
$s_i=1$ for $i=1,\dots,m+k$.

\begin{figure}[!ht]
\centering
\captionsetup{justification=centering}
\includegraphics[scale=0.4]{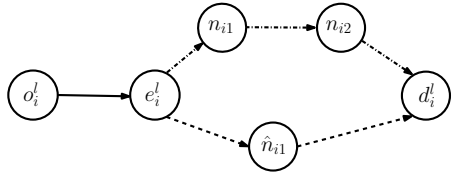}
\begin{center}
\vspace*{-1em}
\caption{The two disjoint literal paths for literal flow $i$ with $\text{Z}_i=3$ and $\hat{\text{Z}}_i=2$.}
\label{fig:literalpaths}
\vspace*{-1em}
\end{center}
\end{figure}

For literal flow $i$, the origin and destination nodes are denoted by
$o_i^l$ and $d_i^l$, respectively. There are exactly two possible
paths for literal flow $i$, referred to as literal paths. See
\figurename~\ref{fig:literalpaths}.
Arc $(o_i^l,e_i^{l})$ appears in both paths, and shown
by the solid line. The first literal path corresponds to variable $n_i$, and
consists of the solid arc and dash-dotted arcs. The first and last
nodes of this path are $o_i^l$ and $d_i^l$, respectively. The
remaining $\text{Z}_i$ nodes between them are
$e_i^{l},n_{i1},\dots,n_{i(\text{Z}_i-1)}$. Thus, there are
$\text{Z}_i+1$ arcs in the path where the first arc is the common arc
$(o_i^l,e_i^{l})$ and each of the remaining arcs represents an
occurrence of variable $n_i$ in the clauses. The second path follows a
similar construction and consists of arc $(o_i^l,e_i^{l})$ and dashed
arcs. The nodes between the origin and destination are
$e_i^{l},\hat{n}_{i1},\dots,\hat{n}_{i(\hat{\text{Z}}_i-1)}$, giving
$\hat{\text{Z}}_i+1$ arcs. The arcs, except the first one, represent
the $\hat{\text{Z}}_i$ occurrences of $\hat{n}_i$ in the clauses. The
above construction is repeated for every literal flow, without any
overlap between the elements defined for a literal flow and
those defined for any other literal flow.

For clause flow $j$, we introduce origin node $o_j^c$ and destination
node $d_j^c$. There are three possible paths referred to as the clause
paths of $j$. Each path corresponds to a literal in the
3-SAT clause and consists of exactly four arcs. The first arc is
common among the three paths and originates from $o_j^c$ to
$e_j^c$. The second arc is from $e_j^c$ to head of the literal arc
defined earlier to represent the occurrence of the literal in this
clause. The third arc is the literal arc itself, and the fourth arc originates
from the tail of the literal arc to $d_j^c$. See
\figurename~\ref{fig:clausepaths} for an illustration, assuming that
clause $j$ contains literal $n_i$, and this is the second occurrence
of the literal. For the other two literals in the clause, two
additional paths are defined similarly. This construction applies to
all the clauses. Hence each clause path has exactly one arc in
common with a literal path.

\begin{figure}[htbp]
\centering
\includegraphics[scale=0.4]{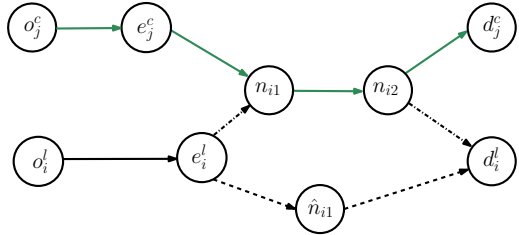}
\begin{center}
\vspace*{-1em}
\caption{One of the three clause paths for clause flow $c_j$.
The corresponding clause is $c_j$, $n_i$ is one of its literals and this is the second
occurrence of the literal.}
\label{fig:clausepaths}
\vspace*{-1em}
\end{center}
\end{figure}

Assume there is a feasible solution for the 3-SAT instance.  For any
variable $n_i$, if it has value true, we route literal flow $i$ on the
literal path defined for $\hat{n}_i$. Otherwise, literal flow $i$ uses
the literal path defined for $n_i$. As a result, if $n_i$ is true, all
the arcs representing the occurrences of $n_i$ are available (i.e.,
has no literal flow) and all the arcs representing the occurrences of
$\hat{n}_i$ are in use by flow $i$. Consider any clause $c_j$ and
suppose it is satisfied by $n_i$ i.e., $n_i$ is true. We route clause
flow $j$ by the clause path defined for this occurrence of $n_i$. This
path is available because the first, second and fourth arcs are always
available to clause flow $j$ as they are specifically defined for $j$,
and the third arc is available since $n_i$ has value true. Using the
resulting routing solution and scheduling this routing for one time
unit yields a feasible solution.

Assume there is a solution for the IFDP instance.  Note that the
deadline for all flows is one time unit, and exactly one unit of
capacity can be used on any arc. Moreover, by construction, the
end-to-end rate of any flow cannot exceed one. These together imply
that the end-to-end rate must equal one at any time point of the IFDP
solution.  Therefore, at any time point, one of the two literal paths
is in use for each literal flow $i$, and one of the three clause paths
is in use by each clause flow $j$.  Thus it is sufficient to consider
the flow routing at any time point in the solution to the IFDP
instance. For any literal flow $i$, if the routing path is the one
corresponding to $n_i$, variable $n_i$ is assigned with value false,
otherwise, it is assigned with value true. This leads to a complete
value assignment for the 3-SAT instance. Now consider any clause flow
$j$, because its end-to-end rate equals one, at least one of the
clause paths is available, meaning that the corresponding literal arc
is not used by the corresponding literal flow, and hence the clause of
the 3-SAT instance is satisfied.  Consequently the answer to the 3-SAT
instance is yes. The proof is then complete by noting
that the reduction is clearly
polynomial.

\section{}
\label{NPhard2}

As for the proof of Theorem \ref{th:NP1}, we construct a reduction
from 3-SAT as follows. There are $F=2m+k$ flows referred to as literal
and clause flows respectively. For convenience, notation for literal and clause
of 3-SAT
reused for the corresponding flows. Thus
$n_1,\hat{n}_1,\dots,n_m,\hat{n}_m$ denote the $2m$ literal flows and
$c_1,\dots,c_k$ denote the $k$ clause flows. The deadlines of literal
and clause flows equal two and one time units, respectively, i.e.,
$t_i=2$, for $i=1,\dots,2m$, and $t_i=1$, for $i=2m+1,\dots,2m+k$. The
size of each flow is one, i.e., $s_i=1$ for $i=1,\dots,2m+k$.

For clause flow $c_j$, we use $o_j^c$ and $d_j^c$ to denote the origin
and destination nodes, respectively.  We define an arc from $o_j^c$ to
an intermediate node $e_j^c$ with capacity three and an arc from
$e_j^c$ to destination node $d_j^c$ with capacity one. This path is
referred to as clause path $c_j$, consisting of two arcs,
see Figure \ref{fig:clausepath_np2}.

\begin{figure}[htbp]
\centering
\captionsetup{justification=centering}
\includegraphics[scale=0.4]{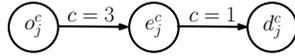}
\begin{center}
\vspace*{-1em}
\caption{Clause path $c_j$.}
\label{fig:clausepath_np2}
\vspace*{-2em}
\end{center}
\end{figure}

For each pair of literal flows $n_i$ and $\hat{n}_i$, we define a set
of four nodes. Two of these nodes are denoted by $o_{i}^l$ and
$\hat{o}_i^l$ representing the origins of the two flows,
respectively. The other two nodes are denoted by $b_i$ and
$b^\prime_i$ for the sake of reference. We define an arc from each of the
two origins $o_{i}^l$ and $\hat{o}_i^l$ to node $b_i$, and an arc from
$b_i$ to $b^\prime_i$. Arc $(b_i, b^\prime_i)$ is hence the bottleneck arc for
flows $n_i$ and $\hat{n}_i$. These three arcs all have capacity one.
Thus, at most one of the two can be routed at a time.  See
\figurename~\ref{fig:4nodes}. All literal flows share a common
destination, denoted by $d^l$.

For each literal flow, we designate one single path for routing,
referred to as literal path. Specifically, for literal flow $n_i$, the
path consists of arcs $(o_i^l, b_i)$, $(b_i, b^\prime_i)$, followed by
traversing through the first arc of all clause paths $c_j$ where $c_j$
$\in$ $\hat{\mathcal{Z}}_i$, and finally to destination node
$d^l$. For this purpose, we need to define some additional arcs as
follows. Denote by $c_{j_1},\dots,c_{j_{\hat{\text{Z}}_i}}$ the clauses
in $\hat{\mathcal{Z}}_i$.  First, we define an arc from $b^\prime_i$
to the first node of clause path $c_{j_1}$, i.e., origin $o_{j_1}^c$.
Next, we define one arc from $e_{j_h}^c$ which is the middle node of clause
path for $j_h$, to $o_{j_{(h+1)}}^c$ which is the origin node of clause path
for $j_{h+1}$, for $h=1,\dots,\hat{\text{Z}}_i-1$. Finally,
we define an arc $o_{j_{\hat{\text{Z}}_i}}^c$ to the destination of
literal flow $d^l$. The capacity of all these new arcs
equals one. For literal flow $\hat{n}_i$, similar construction applies
to ${\mathcal{Z}}_i$.

\begin{figure}[htbp]
\centering
\captionsetup{justification=centering}
\includegraphics[scale=0.4]{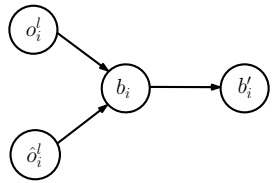}
\begin{center}
\vspace*{-1em}
\caption{The set of four nodes defined for pair $n_i$ and $\hat{n}_i.$ }
\label{fig:4nodes}
\vspace*{-2em}
\end{center}
\end{figure}

By the constructions above, each flow has one single path in the
network, i.e., routing is fully fixed. Moreover, the reduction is
clearly polynomial. We make the following observations.  First, the
literal flows of each pair are mutually exclusive, that is, we can
either schedule $n_i$ or $\hat{n}_i$, but not both
simultaneously. Second, scheduling literal flow $n_i$ has two effects.
First, one capacity unit becomes occupied on the first arc of each
clause path for which the clause is in $\hat{\mathcal{Z}}_i$.  Second,
the first arc of each clause path for which the clause is in
$\mathcal{Z}_i$ has at least one capacity unit available. Scheduling
$\hat{n}_i$ has the opposite effects.  Having these observations in
mind, in the following we show the equivalence of the 3-SAT instance
and the IFDP instance in terms of feasibility.

Suppose there is a yes-answer to the 3-SAT instance. For each pair of
literals $n_i$ and $\hat{n}_i$, we schedule the literal flow with
value true in the first time unit, and the other in the second time
unit. By doing so, all the literal flows are delivered within the
deadline of two time units. Consider any clause $c_j$, at least one of
the literals of this clause holds true.  Therefore, at least one unit
of capacity is available on the first arc of clause path $c_j$.  This,
together with the fact that the second arc of clause path $c_j$ is
defined only for this clause, implies that the clause flow can be
scheduled and delivered within the first time unit.

Conversely, assume we have a feasible solution for the IFDP
instance. Note that the deadline for the clause flows is one time
unit, and for any clause flow, its end-to-end rate can be at most one.
Therefore, the end-to-end rate of a clause flow must equal one
throughout the entire time line of the IFDP scheduling
solution.  Thus, no matter the time point taken, for the first arc of
any clause path, at least one capacity unit is available.  Now
consider any time point of the IFDP solution.  For each pair
of literal flows, if $n_i$ is scheduled, we assign value true to
variable $n_i$. Otherwise we set $n_i$ to be false.  This gives a
true/false assignment of the 3-SAT instance. For the value assignment,
as least one literal of each clause holds true.  This is because at
least one of the three corresponding literal flows is scheduled at the
time moment, as otherwise no capacity of the first arc of the clause
path would be available to the clause flow.  Hence the value
assignment of the 3-SAT instance make all clauses satisfied.  We
remark that, in general, the flow scheduling solution of the
IFDP instance may change over time, however the feasibility
of the IFDP implies that the solution of any time point gives
a feasible solution to the 3-SAT instance.  We can now conclude that
the recognition version of IFDP is NP-complete and its
optimization version is NP-hard.

\section{}
\label{Poly}

Consider an optimal solution to IFDP, and suppose there exists a time
interval of length $\delta$, in which $F'>1$ flows are scheduled and
routed through the common arc.  Without loss of generality, suppose
the flow indexes are $1, \dots, F'$.  Moreover, denote by $\mu_f$ the
amount of capacity of the arc, or equivalently, the end-to-end rate,
allocated to flow $f, f \in \{1, \dots F'\}$.  Hence the amount of demand
delivered in this time interval is $\delta \mu_f, f \in \{1, \dots F'\}$. Note
that none of these flows' deadlines is before the end of the time
interval.

Consider replacing the scheduling solution for the time interval as
follows.  The amount of arc capacity $c = \sum_{f=1}^{F'} \mu_f$ is
allocated to one flow at a time, and this is also the end-to-end rate
of the flow. The order of the flows can be arbitrary. This change is
feasible because the arc is the only common one of all flows, and each
the arc is also the bottleneck for each individual flow.  Moreover,
the scheduling time of flow $f$ is set to $\delta
\mu_f / c$. After the update, for any flow $f$, the amount of demand
delivered remains $\delta \mu_f$ as before.  In addition, the total
scheduling time is $\delta \sum_{f=1}^{F'} \mu_f / c = \delta$, i.e.,
the length of the time interval. Hence the updated solution has no
impact on the overall completion time nor solution feasibility in
meeting the deadlines. Thus the new solution remains optimal.

Applying the above to all time intervals in which multiple flows are
scheduled, we obtain an optimal solution in which one single flow is
scheduled at any time point. At this stage, suppose any two flows $f'$
and $f$ are scheduled consecutively, with $t_f < t_{f'}$. Swapping the
two flows in the schedule with their respective time durations obviously
will not affect the feasibility or optimality of the solution.
Doing so repeatedly if necessary, we eventually obtain an
optimal solution in which the flows are scheduled individually in ascending
order of the deadlines.

\end{appendices}
\section*{Acknowledgement}
This work has been partially supported by the Swedish Research Council.

\bibliographystyle{IEEEtran}
\bibliography{ref}

%\bibliographystyle{IEEEtran}
%\bibliography{ForIEEEBib}% your bib database
\end{document}